\documentclass[twocolumn]{autartarxiv}

\usepackage{amsfonts}
\usepackage{anyfontsize}
\usepackage{amsmath,amssymb}
\usepackage{graphicx}
\usepackage{float}
\usepackage{multicol}
\usepackage{mathtools}
\usepackage{arydshln}
\usepackage{enumitem}
\usepackage[export]{adjustbox}
\usepackage[breaklinks=true]{hyperref}
\usepackage{algorithm}
\usepackage{algpseudocode}
\usepackage{cuted}
\usepackage[title]{appendix}
\usepackage{natbib}
\usepackage{epstopdf}

\newtheorem{theorem}{Theorem}
\newtheorem{lemma}{Lemma}

\newenvironment{proof}[1][Proof]{\noindent\textbf{#1.}\quad}{\hfill$\square$}
\newtheorem{rmk}{Remark}

\begin{document}

\begin{frontmatter}
		\title{An Orthogonal Basis Approach to Formation Shape Control (Extended Version) \thanksref{footnoteinfo}}
		
		\thanks[footnoteinfo]{This paper is the extended version of our paper submitted to Automatica with the proof of some lemmas. Corresponding author M.~de~Queiroz. Tel. +1-225-578-8770. }
		
	\author[USA1]{Tairan Liu}\ead{Tairan.Liu@uga.edu},
	\author[USA2]{Marcio de Queiroz}\ead{mdeque1@lsu.edu}
	
	\address[USA1]{School of Electrical and Computer Engineering, University of Georgia,
		Athens, GA 30602, USA}  
	
	\address[USA2]{Department of Mechanical and Industrial Engineering, Louisiana State University,
		Baton Rouge, LA 70803, USA} 
		
		\begin{keyword}
			Formation Control; Orthogonal Basis; Distance; Signed Volume.
		\end{keyword}

	\begin{abstract}
In this paper, we propose a novel approach to the problem of augmenting distance-based formation controllers with a 
secondary constraint for the purpose of preventing 3D formation ambiguities. Specifically, we introduce three 
controlled variables that form an orthogonal space and uniquely characterize a tetrahedron 
formation in 3D. This orthogonal space incorporates constraints on the inter-agent distances and the
signed volume of tetrahedron substructures. The formation is modeled using a directed graph with a leader-follower 
type configuration and single-integrator dynamics. We show that the proposed decentralized formation controller ensures the
\textit{global} asymptotic stability and the local exponential stability of the desired formation for an \textit{n}-agent system with no ambiguities. 
Unlike previous work, this result is achieved without conditions on the tetrahedrons that form the desired formation or on the control 
gains. 
	\end{abstract}
		
	\end{frontmatter}

\section{Introduction}

In formation shape control, a group of interacting mobile agents are
commanded to acquire and maintain a desired geometric pattern in space. A
well-known method for solving this problem involves regulating a set of
inter-agent distances to values prescribed by the desired shape \citep{de2019formation,krick2009stabilisation}. This method, which is commonly
referred to as distance-based formation control \citep{oh2015survey}, has the
main advantage of being implementable in a fully decentralized manner.
However, this advantage comes with the limitation that the set of
inter-agent distances may not uniquely define the formation position and
orientation in space. Mathematically, this nonuniqueness is related to the
existence of multiple equilibrium points in the multi-agent system distance
dynamics. The question then is how do you steer the system away from the
undesired equilibria and towards the equilibrium corresponding to the
desired formation shape (up to translation and rotation).

The above question is partially answered by requiring the formation graph to
be rigid, which imposes a minimum number of distances to be controlled. This
reduces the unwanted equilibrium points\ to formations that are flipped or
reflected versions of the desired shape. Here, the agents' initial
conditions determine whether the formation converges to one that is
isomorphic to the desired formation or to a flipped/reflected formation.
This implies that rigid distance-based controllers are locally stable.

In recent years, some methods have been introduced to address the limitation
of distance-based formation controllers. The common feature of these methods
is the use of an additional controlled variable (or constraint) that is
capable of distinguishing formation ambiguities. In \cite{mou2015target}, an
approach called target-point control was introduced to rule out the
undesirable equilibria in planar formations. The inter-agent distances and
the order of agents were used to calculate the desired position of the
agents, i.e., target position, in a local coordinate frame. However, the
leader and the first follower cannot be collocated at time zero and the
leader's motion needs to satisfy certain conditions. A similar method was
used in \cite{kang2017distance} with the name ``desired order of neighbors''. In \cite{ferreira2016distance},
inter-agent distance and angular constraints were employed to enlarge the
region of attraction to the desired planar formation by a proper choice of
control gains. In \cite{anderson2017formation}, the signed area of a
triangle was used as the second controlled variable, and convergence
analyses were conducted for special cases of 3- and 4-agent planar
formations. In \cite{sugie2018hierarchical}, the authors further explored
the idea of \cite{anderson2017formation} by applying the area constraints to
only a subset of agents and thus extended the method to $n$ agents. However,
the result in \cite{sugie2018hierarchical} required the triangulated
formation to be composed of equilateral triangles. Recently in \cite%
{liu2019further}, we extended the distance/signed area method to directed 2D
formations of $n$ agents and introduced the concept of strong congruency. In
this result, the triangulations were not restricted to equilateral ones;
however, certain triangulation and control gain conditions had to be met to
prove the asymptotic stability of the desired formation. In \cite{cao2019almost}, a specific control gain value in the signed area term was
found that causes the multi-agent system to have one stable equilibrium
point corresponding to the desired formation and some discrete unstable
equilibria. A switching control strategy based on the signed area or edge
angle was proposed in \cite{liu2020switching} to remove the restrictions on
the shape of the desired formation. In \cite{liu2020distance}, we designed a
non-switching, distance/edge angle-based controller that ensures the
almost-global asymptotic stability of the desired formation under certain
conditions on the triangulations. Recently in \cite{jing2020multiagent}, a
formation controller based on angles and the ``sign'' of the triangulated framework was shown to guarantee almost-global convergence of the angle errors.

For the 3D formation problem, relatively few results exist. For example, 
\cite{ferreira2016adaptive} extended the method in \cite%
{ferreira2016distance} to 3D by using distance and volume constraints.
However, unless the control gains satisfy a persistency of excitation-type
condition, the system under the control of \cite{ferreira2016adaptive} may
still converge to an undesired formation shape. In \cite{lan2018adaptive},
volume constraints were applied to a 4-agent system to distinguish the two
possible orientations of a tetrahedron under the assumption that three of
the agents are at their desired distances.

In this paper, we address the problem of using additional feedback variables
in the distance-based 3D formation controller. Specifically, we introduce a
new method called the \textit{orthogonal basis approach} which decomposes
the feedback variables and control inputs into three orthogonal subspaces.
By applying this decomposition to directed frameworks formed by
tetrahedrons, we are able to guarantee the \textit{global} asymptotic
stability of the desired formation. Moreover, we can show the desired
formation is locally exponentially stable which provides robustness to the
system. These results are achieved with no limitations on the
``tetrahedralizations'' of the desired
formation, control gains, or number of agents. Thus, this work greatly
extends the applicability of the dual-feedback-variable approach for
avoiding 3D formation ambiguities. To the best of our knowledge, it is the
first result to show convergence to the desired 3D formation for all initial
conditions, including collocated and collinear agents. A preliminary version
of our orthogonal basis approach appeared in \cite{liu2020ortho}, where it
was applied to planar formations and achieved almost-global asymptotic
stability of the desired formation.

\section{Background Material \label{sec:back-mat}}

Some background material is reviewed in this section. Throughout the paper,
we use the following vector notation: $x\in \mathbb{R}^{n}$ or $x=\left[
x_{1},\ldots ,x_{n}\right] $ denotes an $n\times 1$ (column) vector, and $x=%
\left[ x_{1},\ldots ,x_{n}\right] $ where $x_{i}\in \mathbb{R}^{m}$ is the
stacked $mn\times 1$ vector.

\subsection{Graph Theory \label{Sec:graph-theory}}

An undirected graph $G$ is represented by a pair $(\mathcal{V},\mathcal{E}%
^{u})$, where $\mathcal{V}=\{1,2,\ldots ,N\}$ is the set of vertices and $%
\mathcal{E}^{u}=\{(i,j)|\,i,j\in \mathcal{V},i\neq j\}\subset \mathcal{V}%
\times \mathcal{V}$ is the set of undirected edges. A directed graph $G$ is
a pair $(\mathcal{V},\mathcal{E}^{d})$ where the edge set $\mathcal{E}^{d}$
is directed in the sense that if $(i,j)\in \mathcal{E}^{d}$ then $i$ is the
source vertex of the edge and $j$ is the sink vertex. The set of neighbors
of vertex $i\in \mathcal{V}$ is defined as $\mathcal{N}_{i}(\mathcal{E}%
^{d})=\{j\in \mathcal{V}|(i,j)\in \mathcal{E}^{d}\}$. For $i\in \mathcal{V}$%
, the out-degree of $i$ (denoted by $\text{out}(i)$) is the number of edges
in $\mathcal{E}^{d}$ whose source is vertex $i$ and sinks are in $\mathcal{V}%
-\{i\}$. If $p_{i}\in \mathbb{R}^{3}$ is the coordinate of the $i$th vertex
of a 3D graph, then a framework $F$ is defined as the pair $(G,p)$ where $p=%
\left[ p_{1},\ldots ,p_{N}\right] \in \mathbb{R}^{3N}$.

Let the map $\mathcal{T}:\mathbb{R}^{3}\rightarrow \mathbb{R}^{3}$ be such
that $\mathcal{T}(x)=\mathcal{R}x+d$ where $\mathcal{R}\in SO(3)$ and $d\in 
\mathbb{R}^{3}$. A framework $F=\left( G,p\right) $ is rigid in $\mathbb{R}%
^{3}$ if all of its continuous motions satisfy $p_{i}(t)=\mathcal{T}(p_{i})$
for all $i=1,\ldots ,N$ and $\forall t\geq 0$ \citep{asimow1979rigidity,izmestiev2009infinitesimal}. A 3D rigid framework is
minimally rigid if and only if $\left\vert \mathcal{E}^{u}\right\vert =3N-6$ 
\citep{anderson2008rigid}. The edge function of a minimally rigid framework $%
\gamma :\mathbb{R}^{3N}\rightarrow \mathbb{R}^{3N-3(3+1)/2}$ is defined as 
\begin{equation}
\gamma (p)=\left[ \ldots ,\left\Vert p_{i}-p_{j}\right\Vert ,\ldots \right] , \, (i,j) \in \mathcal{E}^{u}  \label{edge function}
\end{equation}%
such that its $l$th component, $\left\Vert p_{i}-p_{j}\right\Vert $, relates
to the $l$th edge of $\mathcal{E}^{u}$ connecting the $i$th and $j$th
vertices. Frameworks $(G,p)$ and $(G,\hat{p})$ are equivalent if $\gamma
(p)=\gamma (\hat{p})$, and are congruent if $\left\Vert
p_{i}-p_{j}\right\Vert =\left\Vert \hat{p}_{i}-\hat{p}_{j}\right\Vert $ for
all distinct vertices $i$ and $j$ in $\mathcal{V}$ \citep{jackson2007notes}.
If rigid frameworks $(G,p)$ and $(G,\hat{p})$ are equivalent but not
congruent, they are flip- or flex-ambiguous \citep{anderson2008rigid}.

Frameworks based on directed graphs are required to be constraint consistent
and persistent to maintain its shape \citep{yu2007three}. A persistent graph
is said to be minimally persistent if no single edge can be removed without
losing persistence. A sufficient condition for a directed graph $\left( 
\mathcal{V},\mathcal{E}^{d}\right) $ in $\mathbb{R}^{3}$ to be constraint
consistent is $\text{out}(i)\leq 3$ for all $i\in \mathcal{V}$ (see Lemma 5
of \cite{yu2007three}). A necessary condition for a graph in $\mathbb{R}^{3}$
to be minimally persistent is $\text{out}(i)\leq 3$ for all $i\in \mathcal{V}
$, while a sufficient condition is minimal rigidity \citep{yu2007three}. A
minimally persistent graph can be constructed by the 3D Henneberg insertion
of type I\footnote{%
As shown in \cite{grasegger2018lower}, the 3D Henneberg insertion of type I
does not allow edge splitting operations \citep{eren2005information} in the
graph construction procedure.} \citep{grasegger2018lower}. This method starts
with three vertices with three directed edges, and grows the graph by
iteratively adding a vertex with three outgoing edges. Henceforth, we refer
to a framework constructed in this manner as a 3D Henneberg framework.

\subsection{Strong Congruency \label{Sec: Strong Congr}}

The concept of congruency defined above can distinguish between two
frameworks that are flip- or flex-ambiguous, but does not capture a third
type of ambiguity --- a reflection of the whole framework. In \cite%
{liu2019further}, we introduced the concept of \textit{strong congruency} to
handle this type of ambiguity in 2D. Specifically, Henneberg frameworks $%
F=(G,p)$ and $\hat{F}=(G,\hat{p})$ are said to be \textit{strongly congruent}
if they are congruent and not reflected versions of each other.

In \cite{liu2019further} (Lemma 2.2), it was shown that the signed area of a
triangular framework in addition to a set of edge lengths can be used to
ensure strong congruency in 2D. In order to extend this concept to 3D, we
will employ the \textit{signed volume of a tetrahedron, }$V:\mathbb{R}%
^{12}\rightarrow \mathbb{R}$ \citep{mallison1935use}: 
\begin{align}
V(p)=& \frac{1}{6}\det \left[ 
\begin{array}{cccc}
1 & 1 & 1 & 1 \\ 
p_{1} & p_{2} & p_{3} & p_{4}%
\end{array}%
\right]  \notag \\
=& - \frac{1}{6}\left( p_{1}-p_{4}\right) ^{\intercal }\left[ \left(
p_{2}-p_{4}\right) \times \left( p_{3}-p_{4}\right) \right]
\label{eq:signed-volume}
\end{align}
where $p=\left[ p_{1},p_{2},p_{3},p_{4}\right] $. If the order of vertices $%
1,2,3$ is counterclockwise (resp., clockwise) from an observer located at
vertex $4$ facing the $1$-$2$-$3$ plane, then (\ref{eq:signed-volume}) is
positive (resp., negative). Moreover, this quantity is zero if any three
vertices are collinear or the four vertices are coplanar. As an example of
the use of the signed volume, all the frameworks in Figure \ref{fig:3d-scgt}
are congruent, but only frameworks (a), (b), and (c) are strongly congruent.
Hereafter, we denote the set of all $3$-dimensional frameworks that are
strongly congruent to framework $F$ by $\text{SCgt}^{3}(F)$. 
\begin{figure}[tbph]
\centering
\adjincludegraphics[scale=0.45, trim={{0.0\width} {0.0\height} {0.0\width}
{0.0\height}},clip]{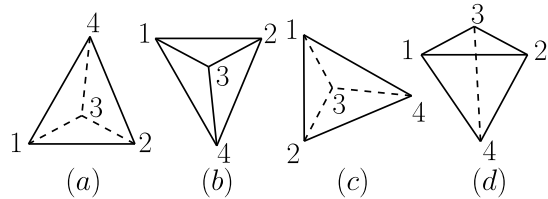}
\caption{All four 3D frameworks are congruent, but only the ones in (a),
(b), and (c) are strongly congruent.}
\label{fig:3d-scgt}
\end{figure}

A 3D Henneberg framework can be divided into tetrahedral sub-frameworks. In
such cases, the signed volume of a 3D Henneberg framework with $N$ vertices
and directed edge set $\mathcal{E}^{d}$, $\mathbf{V}:\mathbb{R}^{3N}\rightarrow \mathbb{R}^{N-3}$, is defined as 
\begin{align}
& \mathbf{V}(p)=\left[ \ldots ,\frac{1}{6}\det \left[ 
\begin{array}{cccc}
1 & 1 & 1 & 1 \\ 
p_{i} & p_{j} & p_{k} & p_{l}%
\end{array}%
\right] ,\ldots \right] ,  \notag \\
& \qquad \forall (l,i),(l,j),(l,k)\in \mathcal{E}^{d}-\{(2,1),(3,1),(3,2)\}
\end{align}%
where $p=\left[p_{1},\ldots ,p_{N}\right]$ and its $n$th component is the
signed volume of the $n$th tetrahedron constructed with vertices $i<j<k<l$.
For example, the signed volume of the framework in Figure \ref%
{fig:signed-volume-framework-example} is given by 
\begin{equation}
\mathbf{V}(p)=\left[ 
\begin{array}{l}
- \frac{1}{6}(p_{1}-p_{4})^{\intercal }\left[ (p_{2}-p_{4})\times
(p_{3}-p_{4}) \right] \\ 
- \frac{1}{6}(p_{1}-p_{5})^{\intercal }\left[ (p_{3}-p_{5})\times
(p_{4}-p_{5}) \right]%
\end{array}%
\right]  \label{Psi}
\end{equation}%
where the first element is positive and the second one negative. Note that
if $\hat{F}$ is a reflected version of $F$, then $\mathbf{V}(p)=-\mathbf{V}(%
\hat{p})$. 
\begin{figure}[tbph]
\centering
\includegraphics[width=0.45\linewidth]{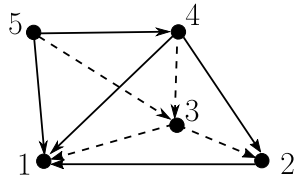}
\caption{Framework with two tetrahedrons.}
\label{fig:signed-volume-framework-example}
\end{figure}

\begin{lemma}
\label{lem:scgt-3d}3D Henneberg frameworks $F=(G,p)$ and $\hat{F}=(G, \hat{p}) $ are strongly congruent if and only if they are equivalent and $\mathbf{V}
(p)=\mathbf{V}(\hat{p})$. (See Appendix \ref{Sec:proof:scgt-3d} for proof.)
\end{lemma}

\subsection{Stability Results}

Here, we recall some results concerning the stability of nonlinear systems.

\begin{lemma}
\label{lem:global-ISS}\citep{khalil2015nonlinear} Suppose $f(x,u)$ is
continuously differentiable and globally Lipschitz in $\left[ x,u\right] $.
If $\dot{x}=f(x,0)$ has a globally exponentially stable (GES) equilibrium
point at the origin, then the system $\dot{x}=f\left( x,u\right) $ is
input-to-state stable (ISS).
\end{lemma}

\begin{lemma}
\label{lem:global-interconn}\citep{khalil2015nonlinear} If the systems $\dot{\eta}=f_{1}(\eta ,\xi )$ and $\dot{\xi}=f_{2}(\xi ,u)$ are ISS, then the
cascade connection 
\begin{equation}
\dot{\eta}=f_{1}(\eta ,\xi ),\quad \dot{\xi}=f_{2}(\xi ,u)
\end{equation}%
is ISS. Consequently, if $\dot{\eta}=f_{1}(\eta ,\xi )$ is ISS and the
origin of $\dot{\xi}=f_{2}(\xi ,0)$ is globally asymptotically stable (GAS),
then the origin of the cascade connection 
\begin{equation}
\dot{\eta}=f_{1}(\eta ,\xi ),\quad \dot{\xi}=f_{2}(\xi ,0)
\end{equation}
is GAS.
\end{lemma}

\section{Problem Statement \label{Probl Stat}}

Consider a system of mobile agents described by directed framework $%
F(t)=(G,p(t))$ where $G=(\mathcal{V},\mathcal{E})$, $\left\vert \mathcal{V}%
\right\vert =N$, $\left\vert \mathcal{E}\right\vert =3N-6$, $N\geq 4$, $p=%
\left[ p_{1},\ldots ,p_{N}\right] $, and $p_{i}\in \mathbb{R}^{3}$ is the
position of agent $i$. The directed edge $(j,i)\in \mathcal{E}$ means that
agent $j$ can measure its relative position to agent $i$, $%
p_{ji}:=p_{i}-p_{j}$, but not the opposite. We assume agent $j$ can sense
all relative positions $p_{ji}$ where $i\in \mathcal{N}_{j}(\mathcal{E})$.
The agents are assumed to be governed by the dynamics 
\begin{equation}
\dot{p}_{i}=u_{i},\quad \forall i\in \mathcal{V}  \label{SI model}
\end{equation}%
where $u_{i}\in \mathbb{R}^{3}$ is the control input.

The desired formation is characterized by a set of desired distances $d_{ji}$%
, $(j,i)\in \mathcal{E}$ and a set of desired signed volumes $V_{ijkl}^{\ast
}=V(p_{i}^{\ast },p_{j}^{\ast },p_{k}^{\ast },p_{l}^{\ast })$, $%
(l,i),(l,j),(l,k)\in \mathcal{E}$ (see (\ref{eq:signed-volume})) where $%
p_{i}^{\ast }$ is the desired position of agent $i$.\footnote{%
Note that $p_{i}^{\ast }$ is not explicitly used by the control since we are
not controlling the global position of the agents. This variable is only
mentioned so we can formally define the desired formation.\bigskip} This
gives the desired framework $F^{\ast }=\left( G,p^{\ast }\right) $ where $%
p^{\ast }=\left[ p_{1}^{\ast },\ldots ,p_{N}^{\ast }\right] $ and $%
\left\Vert p_{j}^{\ast }-p_{i}^{\ast }\right\Vert =d_{ji}$. We assume $%
F^{\ast }$ satisfies the following conditions: i) $F^{\ast }$ is
non-degenerate in 3D space, i.e., all tetrahedrons have nonzero volume; ii)
out$(i)=i-1$ for $i=1,2,3$ and out$(i)=3$ for $i=4,\ldots ,N$; and iii)%
\textbf{\ }if there is an edge between agents $i$ and $j$, the direction
must be $i\leftarrow j$ if $i<j$. Given these conditions, we say that agent
1 is the leader, agent 2 is the first follower, agent 3 is the secondary
follower, and agents $i\geq 4$ are ordinary followers. Note that this
nomenclature is the 3D extension of the leader-first-follower, minimally
persistent framework discussed in \cite{summers2011control}, which was
called an acyclic minimally structural persistent framework in \cite%
{lan2018adaptive}.

Our control objective is to design $u_{i}$, $\forall i\in \mathcal{V}$ such
that 
\begin{equation}
F(t)\rightarrow \text{SCgt}^{3}(F^{\ast })\text{ as }t\rightarrow \infty
\label{objective}
\end{equation}%
for the largest set of initial conditions possible.

\section{Orthogonal Basis}

Ambiguous frameworks in 3D can be discerned by employing the signed volume
of the framework in the formation controller. However, this variable may
introduce new undesired equilibria since the distance and volume constraints
will interfere with each other at certain agent positions. In other words,
these two variables do not always constitute an orthogonal space. To remedy
this situation, we will introduce projection variables that are always
orthogonal.

\begin{figure}[tbph]
\centering
\includegraphics[width=0.5\linewidth]{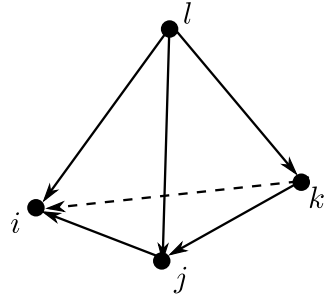}
\caption{Tetrahedron framework.}
\label{fig:3d-framework}
\end{figure}

Consider the tetrahedron in Figure \ref{fig:3d-framework}. If $\left\Vert
p_{ji}\right\Vert \neq 0$ and if there exists a vector $n_{ijk}$ such that 
\begin{equation}
n_{ijk}^{\intercal }p_{ki}=0\quad \text{and}\quad n_{ijk}^{\intercal
}p_{kj}=0,
\end{equation}%
then $p_{ji}^{\intercal }n_{ijk}=0$, $p_{ji}^{\intercal }\left(
n_{ijk}\times p_{ji}\right) =0$, $n_{ijk}^{\intercal }\left( n_{ijk}\times
p_{ji}\right) =0$, and we say $\{p_{ji},n_{ijk}\times p_{ji},n_{ijk}\}$ is
the \textit{orthogonal basis of vertex }$l$. With this in mind, the
projection variables for vertex $l$ are defined as 
\begin{subequations}
\label{eq:3D-OB:projections}
\begin{align}
\zeta _{l}=& p_{li}^{\intercal }p_{ji}=\left\Vert p_{li}\right\Vert
^{2}-p_{li}^{\intercal }p_{lj}  \label{eq:3D-OB:zeta-l} \\
\varphi _{l}=& p_{li}^{\intercal }(n_{ijk}\times p_{ji})=p_{lj}^{\intercal
}(n_{ijk}\times p_{ji})  \label{eq:3D-OB:varphi-l} \\
\vartheta _{l}=& p_{li}^{\intercal }n_{ijk}=p_{lj}^{\intercal
}n_{ijk}=p_{lk}^{\intercal }n_{ijk},  \label{eq:3D-OB:vartheta-l}
\end{align}
\end{subequations}
where 
\begin{equation}
n_{ijk}=\left\{ 
\begin{array}{ll}
p_{ki}\times p_{kj}, & \text{ if }\{i,j,k\}\neq \{1,2,3\} \\ 
&  \\ 
\dfrac{p_{31}\times p_{32}}{\left\Vert p_{31}\times p_{32}\right\Vert }, & 
\text{ if }\{i,j,k\}=\{1,2,3\} \\ 
& \text{ and }\left\Vert p_{31}(t)\times p_{32}(t)\right\Vert \neq 0 \\ 
&  \\ 
n_{123}^{+}, & \text{ if }\{i,j,k\}=\{1,2,3\} \\ 
& \text{ and }\left\Vert p_{31}(0)\times p_{32}(0)\right\Vert =0%
\end{array}%
\right.  \label{eq:n}
\end{equation}%
and $n_{123}^{+}$ is any unit vector satisfying $p_{31}^{\intercal
}n_{123}^{+}=0$ and $p_{32}^{\intercal }n_{123}^{+}=0$. This unit vector can
be computed using Algorithm \ref{alg:pick-n123} below. Notice that (\ref%
{eq:3D-OB:zeta-l}) (resp., (\ref{eq:3D-OB:varphi-l}); (\ref%
{eq:3D-OB:vartheta-l})) is associated with the projection of $p_{li}$ onto
the direction of $p_{ji}$ (resp., $n_{ijk}\times p_{ji}$; $n_{ijk}$). The
reason for differentiating $\{i,j,k\}=\{1,2,3\}$ from the case $%
\{i,j,k\}\neq \{1,2,3\}$ in (\ref{eq:n}) is that our framework is composed
of multiple tetrahedrons where vertices $\{1,2,3\}$ have different
out-degree properties than the others (see Section \ref{Probl Stat}).

\begin{algorithm}
	\begin{algorithmic}[1]
		\State Input: $p_{31} = \left[ x_1, y_1, z_1 \right], p_{32} = \left[ x_2, y_2, z_2 \right]$
		\State Output: $n_{123}$
		
		\State $\text{eps} \gets 1e-3$ 
		
		\If{$\left\Vert p_{31} \right\Vert > \text{eps} $}
		
		\If{$\left\vert z_1 \right\vert > \text{eps} $}
		
		\State $n_z \gets - \left( x_1 + y_1 \right)/z_1$ \;
		\State $n_{t} \gets \left[ 1, 1, n_z\right]$ \;
		
		\ElsIf{$\left\vert y_1 \right\vert > \text{eps} $}
		
		\State $n_y \gets - \left( x_1 + z_1 \right)/y_1$ \;
		\State $n_{t} \gets \left[ 1, n_y, 1\right]$ \;
		
		\ElsIf{$\left\vert x_1 \right\vert > \text{eps} $}
		
		\State $n_x \gets - \left( y_1 + z_1 \right)/x_1$ \;
		\State $n_{t} \gets \left[ n_x, 1, 1\right]$ \;
		\EndIf
		
		\ElsIf{$\left\Vert p_{32} \right\Vert > \text{eps} $}
		
		\If{$\left\vert z_2 \right\vert > \text{eps} $}
		
		\State $n_z \gets - \left( x_2 + y_2 \right)/z_2$ \;
		\State $n_{t} \gets \left[ 1, 1, n_z\right]$ \;
		
		\ElsIf{$\left\vert y_2 \right\vert > \text{eps} $}
		
		\State $n_y \gets - \left( x_2 + z_2 \right)/y_2$ \;
		\State $n_{t} \gets \left[ 1, n_y, 1\right]$ \;
		
		\ElsIf{$\left\vert x_2 \right\vert > \text{eps} $}
		
		\State $n_x \gets - \left( y_2 + z_2 \right)/x_2$ \;
		\State $n_{t} \gets \left[ n_x, 1, 1\right]$ \;
		
		\EndIf
		
		\Else
		\State { $n_t \gets \left[0,0,1\right]$ }
		\EndIf
		
		\State $n_{123} \gets n_t/\left\Vert n_t \right\Vert$
		
		\caption{Selecting $n_{123}$ when $\left\Vert p_{31}(0) \times p_{32}(0) \right\Vert = 0$. (The value $1e-3$ below can be replaced with any sufficiently small number.)}
		\label{alg:pick-n123}
	\end{algorithmic}
\end{algorithm}

Generally, the above projection variables are defined for agents $l\geq 4$
since a tetrahedron is composed of four vertices and only agents $l\geq 4$
have out$(l)=3$. Therefore, special definitions for the projection variables
are required for agents $2$ and $3$. For agent $2$, we define 
\begin{equation}
\zeta _{2}=p_{21}^{\intercal }n_{2},  \label{eq:zeta-2}
\end{equation}%
where 
\begin{equation}
n_{2}=\left\{ 
\begin{array}{ll}
\dfrac{p_{21}}{\left\Vert p_{21}\right\Vert }, & \text{ if }\left\Vert
p_{21}(t)\right\Vert \neq 0 \\ 
&  \\ 
n_{2}^{+}, & \text{ if }\left\Vert p_{21}(0)\right\Vert =0%
\end{array}%
\right.  \label{eq:n2}
\end{equation}%
and $n_{2}^{+}$ is any unit vector. Variables $\varphi _{2}$ and $\vartheta
_{2}$ are undefined since out$(2)=1$. For agent $3$, we let $l=k=3$, $i=1$,
and $j=2$ in (\ref{eq:3D-OB:zeta-l}) and (\ref{eq:3D-OB:varphi-l}) to obtain 
\begin{equation}
\zeta _{3}=p_{31}^{\intercal }p_{21} = \left\Vert p_{31} \right\Vert^2 -
p_{31}^{\intercal} p_{32}  \label{eq:zeta-3}
\end{equation}%
and 
\begin{equation}
\varphi _{3}=p_{31}^{\intercal }\left( n_{123}\times p_{21}\right) .
\label{eq:varphi-3}
\end{equation}%
Here, variable $\vartheta _{3}$ is undefined since out$(3)=2$.

\begin{rmk}
One can show that projection variable (\ref{eq:3D-OB:vartheta-l}) is related
to the signed volume of a tetrahedron and the area of a triangle. Consider (%
\ref{eq:3D-OB:vartheta-l}) when $\{i,j,k\}\neq \{1,2,3\}$. It follows from (%
\ref{eq:n}) that 
\begin{align}
\vartheta _{l}=& p_{li}^{\intercal }n_{ijk}=p_{li}^{\intercal }\left(
p_{ki}\times p_{kj}\right)  \notag \\
=& p_{li}^{\intercal }\left( p_{ki}\times p_{kj}\right) -p_{ki}^{\intercal
}\left( p_{ki}\times p_{kj}\right) =p_{lk}^{\intercal }\left( p_{ki}\times
p_{kj}\right)  \notag \\
=& p_{lk}^{\intercal }\left[ \left( p_{li}-p_{lk}\right) \times \left(
p_{lj}-p_{lk}\right) \right]  \notag \\
=& p_{lk}^{\intercal }\left[ p_{li}\times p_{lj}-p_{li}\times
p_{lk}-p_{lk}\times p_{lj}+p_{lk}\times p_{lk}\right]  \notag \\
=& p_{lk}^{\intercal }\left( p_{li}\times p_{lj}\right) =p_{li}^{\intercal
}\left( p_{lj}\times p_{lk}\right) =-6V_{ijkl}
\label{eq:3D-OB:vartheta-l-volume-not-123}
\end{align}%
where $V_{ijkl}:=V(p_{i},p_{j},p_{k},p_{l})$ from (\ref{eq:signed-volume}).
When $\{i,j,k\}=\{1,2,3\}$ and $\left\Vert p_{31}\times p_{32}\right\Vert
\neq 0$, we have 
\begin{equation}
\vartheta _{l}=p_{l1}^{\intercal }n_{123}=\frac{p_{l1}^{\intercal }\left(
p_{31}\times p_{32}\right) }{\left\Vert p_{31}\times p_{32}\right\Vert }=-%
\frac{3V_{123l}}{\breve{S}_{123}}  \label{eq:3D-OB:vartheta-l-volume-123}
\end{equation}%
where $\breve{S}_{ijk}$ is the regular (unsigned) area of\ $\triangle ijk$
(Heron's formula \citep{zwillinger2002crc}): 
\begin{align}
\breve{S}_{ijk}=& \dfrac{1}{2}\left\Vert p_{ki}\times p_{kj}\right\Vert 
\notag \\
=& \dfrac{1}{4}\left( 2\left\Vert p_{ji}\right\Vert ^{2}\left\Vert
p_{ki}\right\Vert ^{2}+2\left\Vert p_{ji}\right\Vert ^{2}\left\Vert
p_{kj}\right\Vert ^{2}\right.  \notag \\
& \quad \left. +2\left\Vert p_{ki}\right\Vert ^{2}\left\Vert
p_{kj}\right\Vert ^{2}-\left\Vert p_{ji}\right\Vert ^{4}-\left\Vert
p_{ki}\right\Vert ^{4}-\left\Vert p_{kj}\right\Vert ^{4}\right) ^{1/2}.
\label{regular area}
\end{align}
\end{rmk}

The projection variables of a tetrahedralized Henneberg framework $\left( G,
p \right) $ with $\left\vert \mathcal{V}\right\vert =N$ is defined as 
\begin{equation}
\Lambda (p)=\left[ \Lambda _{2},\Lambda _{3},\Lambda _{4}, \ldots ,
\Lambda_{N}\right]
\end{equation}
where 
\begin{equation}
\begin{array}{l}
\Lambda _{2}=\zeta _{2} \\ 
\Lambda _{3}=\left[ \zeta _{3},\varphi _{3}\right] \\ 
\Lambda _{4}=\left[ \zeta _{4},\varphi _{4},\vartheta _{4}\right] \\ 
\multicolumn{1}{c}{\vdots} \\ 
\Lambda _{N}=\left[ \zeta _{N},\varphi _{N},\vartheta _{N}\right] .%
\end{array}
\label{eq:Lambda_i}
\end{equation}

\begin{lemma}
\label{lem:scgt-ortho-3d}3D Henneberg frameworks $F=(G,p)$ and $\hat{F}=(G,%
\hat{p})$ are strongly congruent if and only if $\Lambda (p)=\Lambda (\hat{p}%
)$. (See Appendix \ref{Sec:proof:scgt-ortho-3d} for proof.)
\end{lemma}

Note that by Lemma \ref{lem:scgt-ortho-3d}, the control objective in (\ref%
{objective}) is equivalent to 
\begin{equation}
\Lambda (p(t))\rightarrow \Lambda (p^{\ast })\text{ as }t\rightarrow \infty .
\label{3d-obj}
\end{equation}

\section{Formation Controller \label{Section:OBA:3DForm}}

\subsection{Error Variables}

The control objective will be quantified by the following three \textit{%
projection error} variables 
\begin{subequations}
\label{errors}
\begin{eqnarray}
\tilde{\zeta}_{l} &=&\zeta _{l}-\zeta _{l}^{\ast }  \label{eq:zeta-error-3d}
\\
\tilde{\varphi}_{l} &=&\varphi _{l}-\varphi _{l}^{\ast }
\label{eq:varphi-error-3d} \\
\tilde{\vartheta}_{l} &=&\vartheta _{l}-\vartheta _{l}^{\ast }
\label{eq:vartheta-error-3d}
\end{eqnarray}
\end{subequations}
where the asterisk denotes the desired value for the projection. Since $%
F^{\ast }$ is typically specified in terms of the desired inter-agent
distances, the desired projections can be calculated in terms of $d_{ji}$, $%
(j,i)\in \mathcal{E}$ as shown in Appendix \ref%
{Sec:desired-projection-variables-3d}. The stacked vector of all the
projection errors is represented by $\tilde{\Lambda}=\left[ \tilde{\Lambda}%
_{2},\tilde{\Lambda}_{3},\ldots ,\tilde{\Lambda}_{N}\right] $ where $\tilde{%
\Lambda}_{2}=\tilde{\zeta}_{2}$, $\tilde{\Lambda}_{3}=\left[ \tilde{\zeta}%
_{3},\tilde{\varphi}_{3}\right] $, and $\tilde{\Lambda}_{i}=\left[ \tilde{%
\zeta}_{i},\tilde{\varphi}_{i},\tilde{\vartheta}_{i}\right] $ for $%
i=4,\ldots ,N$.

\begin{lemma}
\label{lem:3d-error-variables} For agent $l$, the planes corresponding to $%
\tilde{\zeta}_{l}=0$, $\tilde{\varphi}_{l}=0$, and $\tilde{\vartheta}_{l}=0$
are mutually orthogonal if $\left\Vert p_{ji}\right\Vert =d_{ji}$, $%
\left\Vert p_{ki}\right\Vert =d_{ki}$, and $\left\Vert p_{kj}\right\Vert
=d_{kj}$, where $i<j<k<l$ and $(l,i),(l,j),(l,k)\in \mathcal{E}$. (See
Appendix \ref{Sec:proof:error-variables} for proof.)
\end{lemma}

\subsection{Control Law}

We propose the following formation control 
\begin{subequations}
\label{ctrl:3D-OB:SI}
\begin{align}
u_{1}=& \ 0  \label{ctrl:3D-OB:SI-1} \\
u_{2}=& \ \mu _{2}\tilde{\zeta}_{2} n_{2}  \label{ctrl:3D-OB:SI-2} \\
u_{3}=& \ \mu _{3}\tilde{\zeta}_{3}\left( p_{31}-p_{32}\right) + \nu _{3} 
\tilde{\varphi}_{3}n_{123}\times \left( p_{31}-p_{32}\right)
\label{ctrl:3D-OB:SI-3} \\
u_{l}=& \ \mu _{l}\tilde{\zeta}_{l}\left( p_{li}-p_{lj}\right) + \nu _{l} 
\tilde{\varphi}_{l}n_{ijk}\times \left( p_{li}-p_{lj}\right)  \notag \\
& + \lambda _{l}\tilde{\vartheta}_{l}n_{ijk},\quad l = 4, \ldots, N
\label{ctrl:3D-OB:SI-l}
\end{align}
\end{subequations}
where $i<j<k<l$, $(l,i),(l,j),(l,k)\in \mathcal{E}$, and $%
\mu_{l},\nu_{l},\lambda _{l}$ are positive control gains.

Note that since $p_{ki}=p_{li}-p_{lk}$ and $p_{kj}=p_{lj}-p_{lk}$ (see
Figure \ref{fig:3d-framework}), the term $n_{ijk}$ defined in (\ref{eq:n})
can be expressed as a function of $p_{li}$, $p_{lj}$, and $p_{lk}$ for $%
(l,i),(l,j),(l,k)\in \mathcal{E}$. Therefore, the above control is
decentralized in the sense that it requires each agent to know its relative
position to neighboring agents only. This means that control (\ref%
{ctrl:3D-OB:SI}) can be implemented in each agent's local coordinate frame.

\begin{rmk}
Since a 3D minimally persistent graph has $3N-6$ edges, a 3D formation
requires $3N-6$ constraints to maintain its shape. Previous 3D formation
control work that utilized distance and volume variables \citep{ferreira2016adaptive} are over-constrained since they require $3N-6$
distance constraints and $N-3$ volume constraints. Although these extra
constraints rule out formation ambiguities, they introduce new undesired
equilibria. On the other hand, the orthogonal basis method proposed here has 
$N-1$ $\zeta $-type projection constraints, $N-2$ $\varphi $-type projection
constraints, and $N-3$ $\vartheta $-type projection constraints, totalling
exactly $3N-6$ constraints.
\end{rmk}

\subsection{Preliminary Analysis}

In preparation for our main results, we present next some preliminary
results.

\begin{lemma}
\label{lem:ndot} With control (\ref{ctrl:3D-OB:SI}), $\dot{n}_{2}=0$ and $%
\dot{n}_{123}=0$.
\end{lemma}

\begin{proof}
If $\left\Vert p_{21}\right\Vert \neq 0$, the time derivative of (\ref{eq:n2}) is given by 
\begin{align}
\dot{n}_{2}=& \dfrac{1}{\left\Vert p_{21}\right\Vert }\left( \dot{p}_{1}-%
\dot{p}_{2}\right) -\dfrac{p_{21}}{\left\Vert p_{21}\right\Vert ^{2}}\dfrac{d%
}{dt}\left\Vert p_{21}\right\Vert  \notag \\
=& \dfrac{1}{\left\Vert p_{21}\right\Vert }\left( \dot{p}_{1}-\dot{p}%
_{2}\right) -\dfrac{p_{21}}{\left\Vert p_{21}\right\Vert ^{2}}\dfrac{%
p_{21}^{\intercal }\left( \dot{p}_{1}-\dot{p}_{2}\right) }{\left\Vert
p_{21}\right\Vert }  \notag \\
=& -\dfrac{1}{\left\Vert p_{21}\right\Vert }u_{2}+\dfrac{p_{21}}{\left\Vert
p_{21}\right\Vert ^{3}}p_{21}^{\intercal }u_{2}.  \label{eq:n2-dot-1}
\end{align}%
After substituting (\ref{ctrl:3D-OB:SI}) into (\ref{eq:n2-dot-1}), we obtain 
\begin{align}
\dot{n}_{2}=& -\dfrac{1}{\left\Vert p_{21}\right\Vert }\left( \mu _{2}\tilde{%
\zeta}_{2}\dfrac{p_{21}}{\left\Vert p_{21}\right\Vert }\right) +\dfrac{p_{21}%
}{\left\Vert p_{21}\right\Vert ^{3}}p_{21}^{\intercal }\left( \mu _{2}\tilde{%
\zeta}_{2}\dfrac{p_{21}}{\left\Vert p_{21}\right\Vert }\right)  \notag \\
=& -\mu _{2}\tilde{\zeta}_{2}\dfrac{p_{21}}{\left\Vert p_{21}\right\Vert ^{2}%
}+\mu _{2}\tilde{\zeta}_{2}\dfrac{p_{21}}{\left\Vert p_{21}\right\Vert ^{2}}
\notag \\
=& \ 0.
\end{align}%
When $\left\Vert p_{21}\right\Vert =0$, it is obvious from (\ref{eq:n2})
that $\dot{n}_{2}=0$.

If $\left\Vert p_{31}\times p_{32}\right\Vert \neq 0$, the derivative of (%
\ref{eq:n}) is given by 
\begin{align}
\dot{n}_{123}=& \frac{1}{\left\Vert p_{31}\times p_{32}\right\Vert }\frac{d}{%
dt}\left( p_{31}\times p_{32}\right)  \notag \\
& -\left( p_{31}\times p_{32}\right) \dfrac{\left( p_{31}\times
p_{32}\right) ^{\intercal }}{\left\Vert p_{31}\times p_{32}\right\Vert ^{3}}%
\dfrac{d}{dt}\left( p_{31}\times p_{32}\right)  \label{eq:3D-OB:n123-dot}
\end{align}%
where 
\begin{equation}
\dfrac{d}{dt}\left( p_{31}\times p_{32}\right) =\left( u_{1}-u_{3}\right)
\times p_{32}+p_{31}\times \left( u_{2}-u_{3}\right) .
\label{eq:3D-OB:p31xp32-dot}
\end{equation}

Substituting (\ref{ctrl:3D-OB:SI}) into (\ref{eq:3D-OB:p31xp32-dot}) yields 
\begin{eqnarray}
&&\dfrac{d}{dt}\left( p_{31}\times p_{32}\right)  \notag \\
&=&-\left( p_{31}-p_{32}\right) \times \left( \mu _{3}\tilde{\zeta}%
_{3}p_{21}+\nu _{3}\tilde{\varphi}_{3}n_{123}\times p_{21}\right)  \notag \\
&&+p_{31}\times \left( \mu _{2}\tilde{\zeta}_{2}n_{2}\right)  \notag \\
&=&-\nu _{3}\tilde{\varphi}_{3}\left\Vert p_{21}\right\Vert ^{2}n_{123}+\mu
_{2}\tilde{\zeta}_{2}p_{31}\times n_{2}.  \label{eq:p31xp32}
\end{eqnarray}

Since $\left\Vert p_{31}\times p_{32}\right\Vert \neq 0$, we have 
\begin{equation}
\left\Vert p_{31}\times p_{32}\right\Vert =\left\Vert p_{31}\times \left(
p_{31}-p_{21}\right) \right\Vert =\left\Vert p_{31}\times p_{21}\right\Vert
\neq 0,  \notag
\end{equation}%
which implies $\left\Vert p_{21}\right\Vert \neq 0$. Then, we can substitute
the first case of (\ref{eq:n2}) and the second case of (\ref{eq:n}) into (%
\ref{eq:p31xp32}) to obtain 
\begin{align}
& \dfrac{d}{dt}\left( p_{31}\times p_{32}\right)  \notag \\
=& -\nu _{3}\tilde{\varphi}_{3}\left\Vert p_{21}\right\Vert ^{2}\dfrac{%
p_{31}\times p_{32}}{\left\Vert p_{31}\times p_{32}\right\Vert }+\mu _{2}%
\tilde{\zeta}_{2}p_{31}\times \dfrac{p_{21}}{\left\Vert p_{21}\right\Vert } 
\notag \\
=& -\nu _{3}\tilde{\varphi}_{3}\left\Vert p_{21}\right\Vert ^{2}\dfrac{%
p_{31}\times p_{32}}{\left\Vert p_{31}\times p_{32}\right\Vert }-\mu _{2}%
\tilde{\zeta}_{2}\dfrac{p_{31}\times p_{32}}{\left\Vert p_{21}\right\Vert } .
\label{eq:p31xp32-2}
\end{align}

Now, substituting (\ref{eq:p31xp32-2}) into (\ref{eq:3D-OB:n123-dot}) gives 
\begin{align}
\dot{n}_{123}=& \dfrac{1}{\left\Vert p_{31}\times p_{32}\right\Vert }  \notag
\\
& \cdot \left( -\nu _{3}\tilde{\varphi}_{3}\left\Vert p_{21}\right\Vert ^{2}%
\dfrac{p_{31}\times p_{32}}{\left\Vert p_{31}\times p_{32}\right\Vert }-\mu
_{2}\tilde{\zeta}_{2}\dfrac{p_{31}\times p_{32}}{\left\Vert
p_{21}\right\Vert }\right)  \notag \\
& -\left( p_{31}\times p_{32}\right) \cdot \dfrac{\left( p_{31}\times
p_{32}\right) ^{\intercal }}{\left\Vert p_{31}\times p_{32}\right\Vert ^{3}}%
\left( -\nu _{3}\tilde{\varphi}_{3}\left\Vert p_{21}\right\Vert ^{2}\right. 
\notag \\
& \left. \cdot \dfrac{p_{31}\times p_{32}}{\left\Vert p_{31}\times
p_{32}\right\Vert }-\mu _{2}\tilde{\zeta}_{2}\dfrac{p_{31}\times p_{32}}{%
\left\Vert p_{21}\right\Vert }\right)  \notag \\
=& \ 0.  \label{eq:3D-OB:n123-dot-equal-0}
\end{align}

For $p_{31}\times p_{32}=0$, it is clear from (\ref{eq:n}) that $\dot{n}%
_{123}=0$. 
\end{proof}

\begin{rmk}
The purpose of $n_{2}^{+}$ in (\ref{eq:n2}) is to force agent $2$ to leave
the collocated initial position with agent $1$. Since $\dot{n}_{2}=0$, the
two cases of (\ref{eq:n2}) have the same value. For example, assume the two
agents are collocated at time zero and let $n_{2}^{+}=\left[ 1,0,0\right] $.
Control (\ref{ctrl:3D-OB:SI-2}) will cause agent $2$ to move away from the
collocated position along vector $\left[ 1,0,0\right] $. Since the agents
are no longer collocated, then $\dfrac{p_{21}}{\left\Vert p_{21}\right\Vert }%
=\left[ 1,0,0\right] $. The vector $n_{123}^{+}$ in (\ref{eq:n}) serves a
similar purpose, viz., to force agent $3$ to leave the collinear initial
position with agents $1$ and $2$. Note that the second and third cases of (%
\ref{eq:n}) have the same value due to $\dot{n}_{123}=0$. If the three
agents are collinear at time zero on the $x$-$y$ plane and $%
n_{123}^{+}=[0,0,1]$ for example, then the second term in (\ref%
{ctrl:3D-OB:SI-3}) will cause agent $3$ to move on the $x$-$y$ plane away
from the collinear position.
\end{rmk}

\begin{lemma}
\label{lem:3D:2-agent-system} For the two-agent system $\{1,2\}$, (\ref%
{ctrl:3D-OB:SI-1}) and (\ref{ctrl:3D-OB:SI-2}) render $\tilde{\zeta}_{2}=0$
GES.
\end{lemma}

\begin{proof}
The time derivative of $\tilde{\zeta}_{2}$ is given by 
\begin{equation}
\dot{\tilde{\zeta}}_{2}=\left( \dot{p}_{1}-\dot{p}_{2}\right) ^{\intercal
}n_{2}+p_{21}^{\intercal }\dot{n}_{2}.  \label{eq:agent2-zeta-error-dot}
\end{equation}%
Applying Lemma \ref{lem:ndot} and substituting (\ref{ctrl:3D-OB:SI-1}) and (%
\ref{ctrl:3D-OB:SI-2}) into (\ref{eq:agent2-zeta-error-dot}), we obtain 
\begin{equation}
\dot{\tilde{\zeta}}_{2}=-n_{2}^{\intercal }u_{2}=-\mu _{2}\tilde{\zeta}%
_{2}\left\Vert n_{2}\right\Vert ^{2}=-\mu _{2}\tilde{\zeta}_{2}
\label{eq:zeta2tilda_dot}
\end{equation}%
which indicates that $\tilde{\zeta}_{2}=0$ is GES.
\end{proof}

Now, consider the Lyapunov function candidates 
\begin{subequations}
\label{W_3D}
\begin{align}
W_{3}& =\dfrac{1}{2}\tilde{\zeta}_{3}^{2}+\dfrac{1}{2}\tilde{\varphi}_{3}^{2}
\label{W_3} \\
W_{l}& =\dfrac{1}{2}\tilde{\zeta}_{l}^{2}+\dfrac{1}{2}\tilde{\varphi}%
_{l}^{2}+\dfrac{1}{2}\tilde{\vartheta}_{l}^{2},\quad l=4,\ldots ,N
\label{eq:3D:Wl}
\end{align}
\end{subequations}
where $i<j<k<l$, $(l,i),(l,j),(l,k)\in \mathcal{E}$.

\begin{lemma}
\label{lem:3D:3-agent-system} If $\left\Vert p_{21}\right\Vert =d_{21}$ and $%
u_{2}=0$, then (\ref{ctrl:3D-OB:SI-1}) and (\ref{ctrl:3D-OB:SI-3}) render $%
\tilde{\Lambda}_{3}=0$ GES for the three-agent system $\{1,2,3\}$.
\end{lemma}

\begin{proof}
The dynamics of $\tilde{\zeta}_{3}$ is given by 
\begin{equation}
\dot{\tilde{\zeta}}_{3}=\dot{p}_{31}^{\intercal }p_{21}+p_{31}^{\intercal }%
\dot{p}_{21}=(u_{1}-u_{3})^{\intercal }p_{21}+p_{31}^{\intercal
}(u_{1}-u_{2})  \label{zeta3_tilda_dot}
\end{equation}%
where (\ref{eq:zeta-error-3d}), (\ref{eq:3D-OB:zeta-l}), and (\ref{SI model}%
) were used. Similarly, the dynamics of $\tilde{\varphi}_{3}$ can be
computed as 
\begin{eqnarray}
\dot{\tilde{\varphi}}_{3} &=&(u_{1}-u_{3})^{\intercal }(n_{123}\times
p_{21})+p_{31}^{\intercal }(\dot{n}_{123}\times p_{21})  \notag \\
&&+p_{31}^{\intercal }(n_{123}\times (u_{1}-u_{2})).
\label{varphi3_tilda_dot}
\end{eqnarray}

Therefore, the time derivative of (\ref{W_3}) is given by 
\begin{align}
\dot{W}_{3}=& \tilde{\zeta}_{3}\dot{\tilde{\zeta}}_{3}+\tilde{\varphi}_{3}%
\dot{\tilde{\varphi}}_{3}  \notag \\
=& \tilde{\zeta}_{3}\left[ p_{21}^{\intercal
}(u_{1}-u_{3})+p_{31}^{\intercal }(u_{1}-u_{2})\right]  \notag \\
& +\tilde{\varphi}_{3}\left[ (n_{123}\times p_{21})^{\intercal }\left(
u_{1}-u_{3}\right) \right.  \notag \\
& \left. +p_{31}^{\intercal }(\dot{n}_{123}\times p_{21})+p_{31}^{\intercal
}(n_{123}\times (u_{1}-u_{2}))\right] .  \label{eq:W3-dot-3D-SI}
\end{align}%
Recall from Lemma \ref{lem:ndot} that $\dot{n}_{123}=0$. Therefore,
substituting $u_{2}=0$, $\left\Vert p_{21}\right\Vert =d_{21}$, (\ref%
{ctrl:3D-OB:SI-1}), and (\ref{ctrl:3D-OB:SI-3}) into (\ref{eq:W3-dot-3D-SI})
gives 
\begin{eqnarray}
\dot{W}_{3} &=&-\left[ \tilde{\zeta}_{3}p_{21}^{\intercal }+\tilde{\varphi}%
_{3}(n_{123}\times p_{21})^{\intercal }\right] u_{3}  \notag \\
&=&-\mu _{3}\tilde{\zeta}_{3}^{2}d_{21}^{2}-\nu _{3}\tilde{\varphi}%
_{3}^{2}\left\Vert n_{123}\times p_{21}\right\Vert ^{2}  \label{W3 dot}
\end{eqnarray}%
where we used the fact that $p_{31}-p_{32}=p_{21}$. Since $\left\Vert
n_{123}\right\Vert =1$, then $\left\Vert n_{123}\times p_{21}\right\Vert
=d_{21}$ and (\ref{W3 dot}) becomes 
\begin{equation}
\dot{W}_{3}=-d_{21}^{2}\left( \mu _{3}\tilde{\zeta}_{3}^{2}+\nu _{3}\tilde{%
\varphi}_{3}^{2}\right) ,
\end{equation}%
which means $\tilde{\Lambda}_{3}=0$ is GES. 
\end{proof}

\begin{lemma}
\label{lem:3D:4-agent-system}If $\left\Vert p_{ji}\right\Vert =d_{ji}$, $%
\left\Vert p_{ki}\right\Vert =d_{ki}$, $\left\Vert p_{kj}\right\Vert =d_{kj}$%
, and $u_{i}=u_{j}=u_{k}=0$, then (\ref{ctrl:3D-OB:SI-l}) ensures that $%
\tilde{\Lambda}_{l}=\left[ \tilde{\zeta}_{l},\tilde{\varphi}_{l},\tilde{%
\vartheta}_{l}\right] =0$ is GES for the tetrahedron formed by agents $%
\{i,j,k,l\}$.
\end{lemma}

\begin{proof}
The time derivative of (\ref{eq:3D:Wl}) along the dynamics of (\ref%
{eq:zeta-error-3d}), (\ref{eq:varphi-error-3d}), and (\ref%
{eq:vartheta-error-3d}) is given by 
\begin{align}
\dot{W}_{l}=& \tilde{\zeta}_{l}\left[ p_{ji}^{\intercal}
(u_{i}-u_{l})+p_{li}^{\intercal }(u_{i}-u_{j})\right]  \notag \\
& +\tilde{\varphi}_{l}\left[ (u_{i}-u_{l})^{\intercal }(n_{ijk}\times
p_{ji})+p_{li}^{\intercal }(\dot{n}_{ijk}\times p_{ji})\right.  \notag \\
& \left. +p_{li}^{\intercal }(n_{ijk}\times (u_{i}-u_{j}))\right]  \notag \\
& +\tilde{\vartheta}_{l}\left[ n_{ijk}^{\intercal
}(u_{i}-u_{l})+p_{li}^{\intercal }\dot{n}_{ijk}\right]
\label{eq:3D-OB:W-l-dot}
\end{align}
where 
\begin{equation}
\dot{n}_{ijk}=(u_{i}-u_{k})\times p_{kj}+p_{ki}\times (u_{j}-u_{k})\text{ if 
}\{i,j,k\} \neq \{1,2,3\},  \notag
\end{equation}%
and $\dot{n}_{123}=0$ from Lemma \ref{lem:ndot}.

After substituting $\left\Vert p_{ji}\right\Vert =d_{ji}$, $%
u_{i}=u_{j}=u_{k}=0$, and (\ref{ctrl:3D-OB:SI-l}) into (\ref%
{eq:3D-OB:W-l-dot}), we obtain 
\begin{align}
\dot{W}_{l}=& - \left[ \tilde{\zeta}_{l}p_{ji}^{\intercal }+\tilde{\varphi}%
_{l}(n_{ijk}\times p_{ji})^{\intercal }+\tilde{\vartheta}_{l}n_{ijk}^{%
\intercal }\right] u_{l}  \notag \\
=& -\mu _{l}d_{ji}^{2}\tilde{\zeta}_{l}^{2}-\nu _{l}\tilde{\varphi}%
_{l}^{2}\left\Vert n_{ijk}\times p_{ji}\right\Vert ^{2}-\lambda _{l}\tilde{%
\vartheta}_{l}^{2}\left\Vert n_{ijk}\right\Vert ^{2}
\label{eq:Wl-dot-3D-SI-2}
\end{align}%
where we used the fact that $p_{ji}^{\intercal }n_{ijk}=0$. Given that $%
\left\Vert p_{ki}\right\Vert =d_{ki}$ and $\left\Vert p_{kj}\right\Vert
=d_{kj}$, we know from (\ref{eq:n}) and (\ref{regular area}) that $%
\left\Vert n_{ijk}\right\Vert $ is constant. Therefore, 
\begin{equation}
\dot{W}_{l}=-\mu _{l}d_{ji}^{2}\tilde{\zeta}_{l}^{2}-\nu _{l}d_{ji}^{2}c%
\tilde{\varphi}_{l}^{2}-\lambda _{l}c\tilde{\vartheta}_{l}^{2}
\label{Wldot1}
\end{equation}%
where $c\ $is some positive constant. It then follows from (\ref{eq:3D:Wl})
and (\ref{Wldot1}) that $\left[ \tilde{\zeta}_{l},\tilde{\varphi}_{l},\tilde{%
\vartheta}_{l}\right] =0$ is GES. 
\end{proof}

\subsection{Main Results}

The following two theorems give our main results.

\begin{theorem}
\label{thm:GAS}Control (\ref{ctrl:3D-OB:SI}) ensures $\tilde{\Lambda}=0$ is
GAS and $F(t)\rightarrow \text{SCgt}^{3}(F^{\ast })$ as $t\rightarrow \infty 
$.
\end{theorem}

\begin{proof}
We know from Lemma \ref{lem:3D:2-agent-system} that the subsystem composed
of agents 1 and 2 is GES at $\tilde{\zeta}_{2}=0$. If a third agent is added
to this subsystem, we get the cascade system 
\begin{subequations}
\label{eq:3D-OB:intercon-3}
\begin{align}
\overset{\cdot }{\tilde{\Lambda}_{3}}& =f_{3}(\tilde{\Lambda}_{3},\tilde{%
\zeta}_{2})  \label{eq:3D-OB:sub-3-xi-dot} \\
\overset{\cdot }{\tilde{\zeta}}_{2}& =g_{2}(\tilde{\zeta}_{2})
\label{eq:3D-OB:sub-3-Xi-dot}
\end{align}
\end{subequations}
where (\ref{eq:3D-OB:sub-3-Xi-dot}) is in fact (\ref{eq:zeta2tilda_dot}). If 
$\tilde{\zeta}_{2}=0$, then $u_{2}=0$ from (\ref{ctrl:3D-OB:SI-2}) and $%
\left\Vert p_{21}\right\Vert =d_{21}$ from (\ref{eq:Lambda_i}) and (\ref%
{eq:zeta-error-3d}). Therefore, (\ref{eq:3D-OB:sub-3-xi-dot}) with $\tilde{%
\zeta}_{2}=0$ is GES at the origin by Lemma \ref{lem:3D:3-agent-system}. It
then follows from Lemma \ref{lem:global-ISS} that (\ref%
{eq:3D-OB:sub-3-xi-dot}) is ISS with respect to input $\tilde{\zeta}_{2}$.
Finally, we can use Lemma \ref{lem:global-interconn} to show that the origin
of (\ref{eq:3D-OB:intercon-3}), i.e., $\left[ \tilde{\zeta}_{2},\tilde{%
\Lambda}_{3}\right] =0$, is GAS.

As we grow the graph step-by-step in the analysis by adding a vertex $l$
with three outgoing edges to any distinct vertices $i$, $j$ and $k$ of the
previous graph, we obtain the following cascade system at each step: 
\begin{subequations}
\label{eq:3D-OB:interconn-l}
\begin{align}
\overset{\cdot }{\tilde{\Lambda}_{l}}=& f_{l}(\tilde{\Lambda}_{l},z_{l-1})
\label{eq:3D-OB:sub-l-xi-dot} \\
\dot{z}_{l-1}=& g_{l-1}(z_{l-1})  \label{eq:3D-OB:sub-l-Xi-dot}
\end{align}
\end{subequations}
where $z_{l-1}=\left[ \tilde{\Lambda}_{2},\ldots ,\tilde{\Lambda}_{l-1}%
\right] $ and $i<j<k<l$. Note that the GAS of $z_{l-1}=0$ for (\ref%
{eq:3D-OB:sub-l-Xi-dot}) was established in the previous step. Therefore, we
only need to check if (\ref{eq:3D-OB:sub-l-xi-dot}) is ISS with respect to
input $z_{l-1}$. If $z_{l-1}=0$, then $u_{i}=u_{j}=u_{k}=0$ from (\ref%
{ctrl:3D-OB:SI-l}). Now, consider subframeworks $F_{l-1}=\left(
G_{l-1},p\right) $ and $F_{l-1}^{\ast }=\left( G_{l-1},p^{\ast }\right) $
where $G_{l-1}$ is the subgraph of $G$ that contains vertices $\{1,\ldots
,l-1\}$ and the corresponding edges connecting these vertices in the
original graph. The condition $z_{l-1}=0$ is equivalent to $\left[ \Lambda
_{2},\ldots ,\Lambda _{l-1}\right] =\left[ \Lambda _{2}^{\ast },\ldots
,\Lambda _{l-1}^{\ast }\right] $ where $\Lambda _{i}^{\ast }=\Lambda _{i}(%
\breve{p}^{\ast })$, and $\breve{p}^{\ast }=\left[ p_{1}^{\ast },\ldots
,p_{l-1}^{\ast }\right] $. This indicates that $F_{l-1}$ and $F_{l-1}^{\ast
} $ are strongly congruent from Lemma \ref{lem:scgt-ortho-3d}. Thus, we know 
$\left\Vert p_{ji}\right\Vert =d_{ji}$, $\left\Vert p_{ki}\right\Vert
=d_{ki} $, and $\left\Vert p_{kj}\right\Vert =d_{kj}$ from Lemma \ref%
{lem:scgt-3d}. We can now use Lemma \ref{lem:3D:4-agent-system} to show that
(\ref{eq:3D-OB:sub-l-xi-dot}) with $z_{l-1}=0$ is GES at the origin. As a
result, (\ref{eq:3D-OB:sub-l-xi-dot}) is ISS by Lemma \ref{lem:global-ISS}.
Finally, we can invoke Lemma \ref{lem:global-interconn} to conclude that $\left[ z_{l-1},\tilde{\Lambda}_{l}\right] =0$ in (\ref{eq:3D-OB:interconn-l}) is GAS.

Repeating this process until $l = N$ leads to the conclusion that $\tilde{%
\Lambda}=0$ is GAS, which implies $\Lambda (p(t)) \rightarrow \Lambda
(p^{\ast})$ as $t \rightarrow \infty$. Thus, we know from Lemma \ref%
{lem:scgt-ortho-3d} that $F(t)\rightarrow \text{SCgt}^{3}(F^{\ast })$ as $t
\rightarrow \infty$. 
\end{proof}

Next, we show that the proposed control yields \textit{local exponential} convergence to the equilibrium point. This property is important in practice since exponential stability is known to provide some level of robustness to the system \citep{khalil2015nonlinear}.

\begin{theorem}
\label{thm:LES} In the neighborhood of $\tilde{\Lambda}=0$, the equilibrium point is locally exponentially stable (LES).
\end{theorem}

\begin{proof}
The error dynamics for $\tilde{\Lambda}$ can be expressed as 
\begin{equation}
\dot{\tilde{\Lambda}}=-A(\tilde{\Lambda})\tilde{\Lambda}
\label{eq:3D:matrix-form}
\end{equation}%
where 
\begin{equation}
A(\tilde{\Lambda})=\left[ 
\begin{array}{cccccccccc}
D_{1} & 0 & 0 & 0 &  & \cdots &  & \cdots &  & 0 \\ 
\star & D_{2} & 0 & 0 &  & \cdots &  & \cdots &  & 0 \\ 
\star & 0 & D_{3} & 0 &  & \cdots &  & \cdots &  & 0 \\ 
\star & \star & \star & D_{4} & 0 & 0 & 0 & \cdots &  & 0 \\ 
\star & \star & \star & 0 & D_{5} & 0 & 0 & \cdots &  & 0 \\ 
\star & \star & \star & 0 & 0 & D_{6} & 0 & \cdots &  & 0 \\ 
\vdots &  &  &  &  &  & \ddots &  &  & \vdots \\ 
\star &  & \cdots &  & \cdots &  & \star & \star & 0 & 0 \\ 
\star &  & \cdots &  & \cdots &  & \star & 0 & \star & 0 \\ 
\star &  & \cdots &  & \cdots &  & \star & 0 & 0 & D_{3N-6}%
\end{array}
\right]  \label{P}
\end{equation}
is a lower triangular $\left[ 3\left( N-3\right) +3\right] \times $ $\left[
3\left( N-3\right) +3\right] $ matrix whose diagonal elements are $D_{1}=\mu
_{2}$, $D_{2}=\mu _{3}\left\Vert p_{21}\right\Vert ^{2}$, $D_{3}=\nu
_{3}\left\Vert n_{123}\times p_{21}\right\Vert ^{2}$, $D_{4}=\mu
_{4}\left\Vert p_{21}\right\Vert ^{2}$, $D_{5}=\nu _{4}\left\Vert
n_{123}\times p_{21}\right\Vert ^{2}$, $D_{6}=\lambda _{4}\left\Vert
n_{123}\right\Vert ^{2}$, $\ldots $, $D_{3N-8}=\mu _{N}\left\Vert
p_{N_{1}N_{2}}\right\Vert ^{2}$, $D_{3N-7}=\nu_{N}\left\Vert
n_{N_{1}N_{2}N_{3}}\times p_{N_{1}N_{2}}\right\Vert ^{2}$, and $%
D_{3N-6}=\lambda _{N}\left\Vert n_{N_{1}N_{2}N_{3}}\right\Vert ^{2}$ where $%
N_{1}$, $N_{2}$, and $N_{3}$ ($N_{1}<N_{2}<N_{3}$) are the out-neighbors of
agent $N$.

Linearizing (\ref{eq:3D:matrix-form}) at the equilibrium $\tilde{\Lambda}=0$
yields 
\begin{equation}
\dot{\tilde{\Lambda}}\approx -A(0)\tilde{\Lambda}
\end{equation}%
where $A(0)$ is a constant matrix whose eigenvalues (diagonal elements) are
positive and can be made arbitrarily large by adjusting the control gains. Therefore, $\tilde{\Lambda}=0$ is LES.
\end{proof}

\begin{rmk}
It is worth noting that the 2D formation problem can be viewed as a
degenerate case of the 3D problem. Specifically, we can consider the
coordinates of agent $i$ as $p_{i}=\left[ x_{i},y_{i},0\right] $ and express
the control law with the third component equal to zero, i.e., $u_{i}=\left[
u_{ix},u_{iy},0\right] $. The results in Theorems \ref{thm:GAS} and \ref%
{thm:LES} are also valid in the 2D case. This analysis is omitted here since
it is based on similar arguments as above.
\end{rmk}

\section{Conclusion}

This paper introduced a new set of controlled variables for avoiding
undesirable equilibria in the 3D distance-based formation control approach.
The proposed variables, which are related to the inter-agent distances,
signed volume of the framework, and areas of the triangular faces, form an
orthogonal basis that decomposes the control inputs into orthogonal
subspaces. The resulting formation controller guarantees the global
asymptotic stability and the local exponential stability of the desired
formation shape with any initial conditions. This result is valid for any
tetrahedralized-like framework with no conditions on the formation shape or
control gains.

The proposed approach can also handle the 2D formation problem, unifying the
two problems. Specifically, by setting the $z$-coordinate of each agent
position to zero, we arrive at the orthogonal basis controller for 2D
formations that appeared in our preliminary result in \cite{liu2020ortho}.

\begin{appendices}
\section{Lemma Proofs}

\subsection{Lemma \protect\ref{lem:scgt-3d}}

\label{Sec:proof:scgt-3d}

Consider Figure \ref{fig:dihedral_angle_directed_height} where $%
n_{123}=p_{23}\times p_{21}$ and $n_{124}=p_{24}\times p_{21}$ are the
vectors normal to planes 1-2-3 and 1-2-4, respectively, and $\alpha $ is the
dihedral angle between the two planes. Then, we have that 
\begin{eqnarray}
\cos \alpha &=&\frac{n_{123}^{\intercal }n_{124}}{\left\Vert
n_{123}\right\Vert \left\Vert n_{124}\right\Vert }  \notag \\
&=&\frac{(p_{23}\times p_{21})^{\intercal }(p_{24}\times p_{21})}{\left\Vert
p_{23}\right\Vert \left\Vert p_{21}\right\Vert \sin \theta _{123}\left\Vert
p_{24}\right\Vert \left\Vert p_{21}\right\Vert \sin \theta _{124}}
\label{eq:cos-alpha1}
\end{eqnarray}%
where $\theta _{ijk}$ is the face angle between edges $(j,i)$ and $(j,k)$.
Using property 
\begin{equation}
(a\times b)^{\intercal }(c\times d)=(a^{\intercal }c)(b^{\intercal
}d)-(a^{\intercal }d)(b^{\intercal }c),
\end{equation}%
we obtain%
\begin{eqnarray}
&&(p_{23}\times p_{21})^{\intercal }(p_{24}\times p_{21})  \notag \\
&=&(p_{23}^{\intercal }p_{24})(p_{21}^{\intercal }p_{21})-(p_{23}^{\intercal
}p_{21})(p_{21}^{\intercal }p_{24})  \notag \\
&=&\left\Vert p_{23}\right\Vert \left\Vert p_{24}\right\Vert \cos \theta
_{423}\left\Vert p_{21}\right\Vert ^{2}  \notag \\
&&-\left\Vert p_{23}\right\Vert \left\Vert p_{21}\right\Vert \cos \theta
_{123}\left\Vert p_{21}\right\Vert \left\Vert p_{24}\right\Vert \cos \theta
_{124}  \notag \\
&=&\left\Vert p_{21}\right\Vert ^{2}\left\Vert p_{23}\right\Vert \left\Vert
p_{24}\right\Vert \left( \cos \theta _{423}-\cos \theta _{123}\cos \theta
_{124}\right) .  \label{eq:cos-alpha-den_1}
\end{eqnarray}%
Substituting (\ref{eq:cos-alpha-den_1}) into (\ref{eq:cos-alpha1}) gives 
\begin{equation}
\cos \alpha =\frac{\cos \theta _{423}-\cos \theta _{123}\cos \theta _{124}}{%
\sin \theta _{123}\sin \theta _{124}}.  \label{eq:cos-dihedral-angle}
\end{equation}

A dihedral angle $\alpha $ can be associated with the signed volume by
defining a \textit{signed dihedral angle}, $\alpha _{s}$. To this end, given
the tetrahedron in Figure \ref{fig:dihedral_angle_directed_height}, we can
define its \textit{signed height}, $h$, to have the same sign as the signed
volume. Then, 
\begin{equation}
\sin \alpha _{s}=\frac{h}{b}  \label{eq:sin-dihedral-angle}
\end{equation}%
where $b$ is the distance from vertex $4$ to edge $(1,2)$. Combining (\ref%
{eq:cos-dihedral-angle}) and (\ref{eq:sin-dihedral-angle}), we can calculate
the signed dihedral angle as 
\begin{equation}
\alpha _{s}=\arctan \text{2}(h/b,\cos \alpha ).  \label{alpha_s}
\end{equation}%
\begin{figure}[tbph]
\centering
\adjincludegraphics[scale=1.1, trim={{0.0\width}
		{0.0\height} {0.0\width} {0.00\height}},clip]{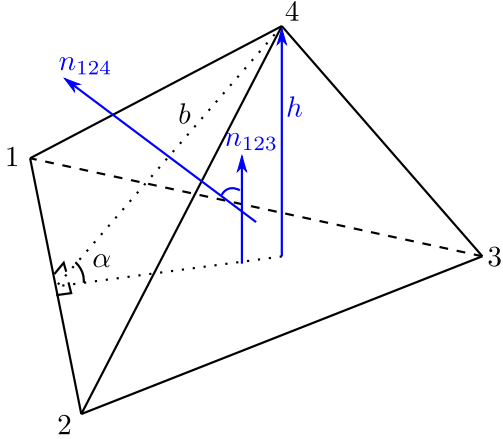}
\caption{Signed dihedral angle and signed height of a tetrahedron.}
\label{fig:dihedral_angle_directed_height}
\end{figure}

Now, we can prove Lemma \ref{lem:scgt-3d} as follows.

\textit{(Proof of $\Rightarrow $)} If $F$ and $\hat{F}$ are strongly
congruent, then $\left\Vert p_{i}-p_{j}\right\Vert =\left\Vert \hat{p}_{i} - 
\hat{p}_{j}\right\Vert $, $\forall i,j\in \mathcal{V}$ and $\mathbf{V}(p) = 
\mathbf{V}(\hat{p})$ by definition. Therefore, since $\mathcal{E}\subset 
\mathcal{V}\times \mathcal{V}$, we know $\left\Vert p_{i}-p_{j}\right\Vert =
\left\Vert \hat{p}_{i}-\hat{p}_{j}\right\Vert $, $\forall (i,j)\in \mathcal{E%
}$, i.e., $F$ and $\hat{F}$ are equivalent.

\textit{(Proof of $\Leftarrow$)} If $\left\vert \mathcal{V} \right\vert = 4$%
, then framework equivalency and congruency are equivalent, so the
conditions for strong congruency are trivially satisfied.

If a vertex is added such that $\left\vert \mathcal{V}\right\vert =5$, the
resulting 3D framework will have three additional edges and one additional
tetrahedron. Consider without loss of generality the framework in Figure \ref%
{fig:tri_bipyramid}, and denote the signed dihedral angle between planes
1-2-3 and 1-2-4 as $\alpha_{s1}$ and between planes 1-2-3 and 1-2-5 as $%
\alpha _{s2}$. Then, the signed dihedral angle between planes 1-2-4 and
1-2-5 is $\alpha _{s3}=\alpha _{s1}+\alpha _{s2}$.

\begin{figure}[tbph]
\centering
\adjincludegraphics[scale=0.9, trim={{0.0\width}
		{0.0\height} {0.0\width} {0.00\height}},clip]{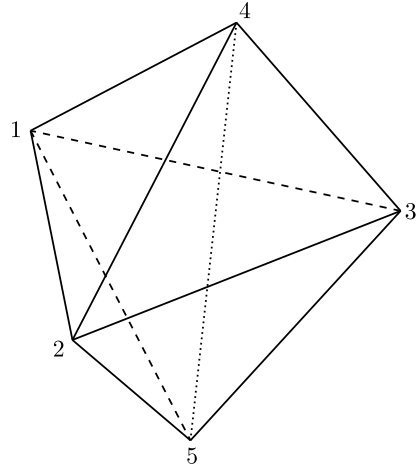}
\caption{Framework with $\left\vert \mathcal{V}\right\vert =5$ containing
two tetrahedrons.}
\label{fig:tri_bipyramid}
\end{figure}

Since $F$ and $\hat{F}$ are equivalent, $\left\Vert p_{i}-p_{j}\right\Vert
=\left\Vert \hat{p}_{i}-\hat{p}_{j}\right\Vert $, $\forall (i,j)\in \mathcal{%
E}$. This along with $\mathbf{V}(p)=\mathbf{V}(\hat{p})$ indicates that all
face angles in $F$ and $\hat{F}$ are equal. From (\ref{alpha_s}), we then
know that $\alpha _{s1}=\hat{\alpha}_{s1}$, $\alpha _{s2}=\hat{\alpha}_{s2}$%
, and $\alpha _{s3}=\hat{\alpha}_{s3}$ where $\hat{\alpha}_{si}$ denotes the 
$i$th dihedral angle of $\hat{F}$. Using (\ref{eq:cos-dihedral-angle}), we
have 
\begin{equation}
\cos \alpha _{s3}=\frac{\cos \theta _{425}-\cos \theta _{124}\cos \theta
_{125}}{\sin \theta _{124}\sin \theta _{125}}.
\end{equation}%
Since $\alpha _{s3}(p)=\alpha _{s3}(\hat{p})$ and all face angles are equal,
we obtain 
\begin{equation}
\cos \theta _{425}=\cos \hat{\theta}_{425}.  \label{cos425}
\end{equation}%
After applying (\ref{cos425}) to the law of cosines, we arrive at $%
\left\Vert p_{5}-p_{4}\right\Vert =\left\Vert \hat{p}_{5}-\hat{p}%
_{4}\right\Vert $. This proves that $\left\Vert p_{i}-p_{j}\right\Vert
=\left\Vert \hat{p}_{i}-\hat{p}_{j}\right\Vert $, $\forall i,j\in \mathcal{V}
$, so $F$ and $\hat{F}$ are strongly congruent for $\left\vert \mathcal{V}%
\right\vert =5$.\footnote{%
This analysis is not restricted to convex structures such as Figure \ref{fig:tri_bipyramid}. A similar analysis can be applied to concave cases.}

As more vertices are added, each new vertex will create a new tetrahedron.
Thus, the above process can be recursively employed to show that $F$ and $\hat{F}$ are strongly congruent for $\left\vert \mathcal{V}\right\vert =N$.

\subsection{Lemma \protect\ref{lem:scgt-ortho-3d}}

\label{Sec:proof:scgt-ortho-3d}

\textit{(Proof of $\Rightarrow $)} From (\ref{eq:Lambda_i}), (\ref%
{eq:3D-OB:zeta-l}), and the fact that strong congruency implies $\left\Vert
p_{ij}\right\Vert =\left\Vert \hat{p}_{ij}\right\Vert $ $\forall i, j \in 
\mathcal{V}$, we have 
\begin{equation}
\zeta _{2}(p)=\left\Vert p_{21}\right\Vert =\left\Vert \hat{p}%
_{21}\right\Vert =\zeta _{2}(\hat{p})  \label{eq:3D-OB:zeta2-equal}
\end{equation}
and 
\begin{align}
\zeta _{3}(p)=& p_{31}^{\intercal }p_{21}=\left\Vert p_{31}\right\Vert
\left\Vert p_{21}\right\Vert \frac{\left\Vert p_{31}\right\Vert
^{2}+\left\Vert p_{21}\right\Vert ^{2}-\left\Vert p_{32}\right\Vert ^{2}}{%
2\left\Vert p_{31}\right\Vert \left\Vert p_{21}\right\Vert }  \notag \\
=& \frac{\left\Vert p_{31}\right\Vert ^{2}+\left\Vert p_{21}\right\Vert
^{2}-\left\Vert p_{32}\right\Vert ^{2}}{2}  \notag \\
=& \frac{\left\Vert \hat{p}_{31}\right\Vert ^{2}+\left\Vert \hat{p}%
_{21}\right\Vert ^{2}-\left\Vert \hat{p}_{32}\right\Vert ^{2}}{2}=\zeta _{3}(%
\hat{p}).  \label{zeta3}
\end{align}
From (\ref{eq:varphi-3}) and (\ref{regular area}), we obtain 
\begin{align}
\varphi _{3}(p)& = n_{123}^{\intercal} \left(p_{31} \times p_{32}\right) =
\left\Vert p_{31} \times p_{32} \right\Vert = 2\breve{S}_{123}(p)  \notag \\
& =\dfrac{1}{2}\left( 2\left\Vert p_{21}\right\Vert ^{2}\left\Vert
p_{32}\right\Vert ^{2}+2\left\Vert p_{31}\right\Vert ^{2}\left\Vert
p_{32}\right\Vert ^{2}-\left\Vert p_{21}\right\Vert ^{4}\right.  \notag \\
& \left. \quad -\left\Vert p_{31}\right\Vert ^{4}-\left\Vert
p_{32}\right\Vert ^{4}+2\left\Vert p_{21}\right\Vert ^{2}\left\Vert
p_{31}\right\Vert ^{2}\right) ^{1/2}  \notag \\
& =\dfrac{1}{2}\left( 2\left\Vert \hat{p}_{21}\right\Vert ^{2}\left\Vert 
\hat{p}_{32}\right\Vert ^{2}+2\left\Vert \hat{p}_{31}\right\Vert
^{2}\left\Vert \hat{p}_{32}\right\Vert ^{2}-\left\Vert \hat{p}%
_{21}\right\Vert ^{4}\right.  \notag \\
& \left. \quad -\left\Vert \hat{p}_{31}\right\Vert ^{4}-\left\Vert \hat{p}%
_{32}\right\Vert ^{4}+2\left\Vert \hat{p}_{21}\right\Vert ^{2}\left\Vert 
\hat{p}_{31}\right\Vert ^{2}\right) ^{1/2}  \notag \\
& =2\breve{S}_{123}(\hat{p})=\varphi _{3}(\hat{p}).
\label{eq:3D-OB:varphi3-equal}
\end{align}

The relation $\zeta _{4}(p)=\zeta _{4}(\hat{p})$ can be shown as in (\ref%
{zeta3}). From (\ref{eq:3D-OB:varphi-l}), (\ref{eq:n}), and (\ref{regular
area}), we get 
\begin{align}
\varphi _{4}(p)& =n_{123}^{\intercal }\left( p_{21}\times p_{41}\right) 
\notag \\
=& \left( \dfrac{p_{31}\times p_{32}}{\left\Vert p_{31}\times
p_{32}\right\Vert }\right) ^{\intercal }\left( p_{21}\times p_{41}\right) 
\notag \\
=& \dfrac{1}{\left\Vert p_{31}\times p_{32}\right\Vert }\left( -p_{32}\times
p_{31}\right) ^{\intercal }\left( -p_{41}\times p_{21}\right)  \notag \\
=& \dfrac{1}{\left\Vert p_{31}\times p_{32}\right\Vert }\left[ \left(
-p_{32}\times p_{32}-p_{32}\times p_{21}\right) ^{\intercal }\left(
-p_{42}\times p_{21}\right. \right.  \notag \\
& \left. \left. -p_{21}\times p_{21}\right) \right]  \notag \\
=& \dfrac{1}{\left\Vert p_{31}\times p_{32}\right\Vert }\left( p_{32}\times
p_{12}\right) ^{\intercal }\left( p_{42}\times p_{12}\right)  \notag \\
=& \dfrac{1}{2\breve{S}_{123}(p)}\left[ (p_{32})^{\intercal
}p_{42}(p_{12})^{\intercal }p_{12}-(p_{32})^{\intercal
}p_{12}(p_{12})^{\intercal }p_{42}\right]  \notag \\
=& \dfrac{1}{2\breve{S}_{123}(p)}\left[ \left\Vert p_{32}\right\Vert
\left\Vert p_{42}\right\Vert \dfrac{\left\Vert p_{32}\right\Vert
^{2}+\left\Vert p_{42}\right\Vert ^{2}-\left\Vert p_{43}\right\Vert ^{2}}{%
2\left\Vert p_{32}\right\Vert \left\Vert p_{42}\right\Vert }\left\Vert
p_{12}\right\Vert ^{2}\right.  \notag \\
& \quad \left. -\left\Vert p_{32}\right\Vert \left\Vert p_{12}\right\Vert 
\dfrac{\left\Vert p_{32}\right\Vert ^{2}+\left\Vert p_{12}\right\Vert
^{2}-\left\Vert p_{13}\right\Vert ^{2}}{2\left\Vert p_{32}\right\Vert
\left\Vert p_{12}\right\Vert }\right.  \notag \\
& \quad \left. \cdot \left\Vert p_{12}\right\Vert \left\Vert
p_{42}\right\Vert \dfrac{\left\Vert p_{12}\right\Vert ^{2}+\left\Vert
p_{42}\right\Vert ^{2}-\left\Vert p_{41}\right\Vert ^{2}}{2\left\Vert
p_{12}\right\Vert \left\Vert p_{42}\right\Vert }\right] .  \label{varphi4}
\end{align}%
Since (\ref{varphi4}) is only a function of the inter-agent distances, then $%
\varphi _{4}(p)=\varphi _{4}(\hat{p})$. A useful formula for calculating the
signed volume $V_{ijkl}$ is given by the Cayley-Menger determinant \citep{sommerville2011introduction}: 
\begin{equation}
V_{ijkl}=\pm \sqrt{\frac{1}{288}\left\vert 
\begin{array}{ccccc}
0 & 1 & 1 & 1 & 1 \\ 
1 & 0 & \left\Vert p_{ji}\right\Vert ^{2} & \left\Vert p_{ki}\right\Vert ^{2}
& \left\Vert p_{li}\right\Vert ^{2} \\ 
1 & \left\Vert p_{ji}\right\Vert ^{2} & 0 & \left\Vert p_{kj}\right\Vert ^{2}
& \left\Vert p_{lj}\right\Vert ^{2} \\ 
1 & \left\Vert p_{ki}\right\Vert ^{2} & \left\Vert p_{kj}\right\Vert ^{2} & 0
& \left\Vert p_{lk}\right\Vert ^{2} \\ 
1 & \left\Vert p_{li}\right\Vert ^{2} & \left\Vert p_{lj}\right\Vert ^{2} & 
\left\Vert p_{lk}\right\Vert ^{2} & 0%
\end{array}%
\right\vert },  \label{eq:V*}
\end{equation}%
where the sign convention described in Section \ref{Sec: Strong Congr}
determines if the sign is positive or negative. From (\ref%
{eq:3D-OB:vartheta-l-volume-123}), we have $\vartheta _{4}(p)= -
3V_{1234}(p)/\breve{S}_{123}(p)$. Given (\ref{regular area}) and (\ref{eq:V*}%
), we can see that $\vartheta _{4}(p)$ is only dependent on the inter-agent
distances and the sign of the volume; hence, it is clear that $\vartheta
_{4}(p)=\vartheta _{4}(\hat{p})$.

Repeating this analysis for $\zeta _{l}$, $\varphi _{l}$, and $\vartheta_{l} 
$, $l=5,\ldots ,N$ leads to the conclusion that $\Lambda (p)=\Lambda (\hat{p}%
)$.

\textit{(Proof of $\Leftarrow $)} If $\Lambda (p)=\Lambda (\hat{p})$, then $%
\zeta _{l}(p)=\zeta _{l}(\hat{p})$, $l=2,\ldots ,N$, $\varphi
_{l}(p)=\varphi _{l}(\hat{p})$, $l=3,\ldots ,N$, and $\vartheta
_{l}(p)=\vartheta _{l}(\hat{p})$, $l=4,\ldots ,N$. From $\zeta _{2}(p)=\zeta
_{2}(\hat{p})$, we obtain 
\begin{equation}
\left\Vert p_{21}\right\Vert =\left\Vert \hat{p}_{21}\right\Vert
\label{eq:3D-OB:p-21-two-frame}
\end{equation}%
where (\ref{eq:Lambda_i}) was used. From $\zeta _{3}(p)=\zeta _{3}(\hat{p})$
and (\ref{eq:3D-OB:zeta-l}), we have 
\begin{align}
& \left\Vert p_{31}\right\Vert \left\Vert p_{21}\right\Vert \frac{\left\Vert
p_{31}\right\Vert ^{2}+\left\Vert p_{21}\right\Vert ^{2}-\left\Vert
p_{32}\right\Vert ^{2}}{2\left\Vert p_{31}\right\Vert \left\Vert
p_{21}\right\Vert }  \notag \\
=& \frac{\left\Vert p_{31}\right\Vert ^{2}+\left\Vert p_{21}\right\Vert
^{2}-\left\Vert p_{32}\right\Vert ^{2}}{2}  \notag \\
=& \frac{\left\Vert \hat{p}_{31}\right\Vert ^{2}+\left\Vert \hat{p}%
_{21}\right\Vert ^{2}-\left\Vert \hat{p}_{32}\right\Vert ^{2}}{2}
\label{eq:3D-OB:zeta3-equal-base}
\end{align}%
where the law of cosines were used. This leads to 
\begin{equation}
\left\Vert p_{31}\right\Vert ^{2}+\left\Vert p_{21}\right\Vert
^{2}-\left\Vert p_{32}\right\Vert ^{2}=\left\Vert \hat{p}_{31}\right\Vert
^{2}+\left\Vert \hat{p}_{21}\right\Vert ^{2}-\left\Vert \hat{p}%
_{32}\right\Vert ^{2}.  \label{eq:3D-OB:zeta3-equal-two-frame}
\end{equation}
From $\varphi _{3}(p)=\varphi _{3}(\hat{p})$ and (\ref{eq:varphi-3}), we get 
\begin{equation}
\breve{S}_{123}(p)=\breve{S}_{123}(\hat{p}).
\label{eq:3D-OB:S-equal-two-frame}
\end{equation}%
Applying (\ref{regular area}) to (\ref{eq:3D-OB:S-equal-two-frame}) yields 
\begin{align}
& 2\left\Vert p_{21}\right\Vert ^{2}\left\Vert p_{31}\right\Vert
^{2}+2\left\Vert p_{21}\right\Vert ^{2}\left\Vert p_{32}\right\Vert
^{2}+2\left\Vert p_{31}\right\Vert ^{2}\left\Vert p_{32}\right\Vert ^{2} 
\notag \\
& \quad -\left\Vert p_{21}\right\Vert ^{4}-\left\Vert p_{31}\right\Vert
^{4}-\left\Vert p_{32}\right\Vert ^{4}  \notag \\
& =2\left\Vert \hat{p}_{21}\right\Vert ^{2}\left\Vert \hat{p}%
_{31}\right\Vert ^{2}+2\left\Vert \hat{p}_{21}\right\Vert ^{2}\left\Vert 
\hat{p}_{32}\right\Vert ^{2}+2\left\Vert \hat{p}_{31}\right\Vert
^{2}\left\Vert \hat{p}_{32}\right\Vert ^{2}  \notag \\
& \quad -\left\Vert \hat{p}_{21}\right\Vert ^{4}-\left\Vert \hat{p}%
_{31}\right\Vert ^{4}-\left\Vert \hat{p}_{32}\right\Vert ^{4}.
\label{eq:3D-OB:varphi4-equal-two-frame}
\end{align}%
Combining (\ref{eq:3D-OB:p-21-two-frame}), (\ref%
{eq:3D-OB:zeta3-equal-two-frame}), and (\ref%
{eq:3D-OB:varphi4-equal-two-frame}) gives $\left\Vert p_{31}\right\Vert
=\left\Vert \hat{p}_{31}\right\Vert $ and $\left\Vert p_{32}\right\Vert
=\left\Vert \hat{p}_{32}\right\Vert $.

Next, from $\zeta _{4}(p)=\zeta _{4}(\hat{p})$, $\varphi _{4}(p)=\varphi
_{4}(\hat{p})$, and $\vartheta _{4}(p)=\vartheta _{4}(\hat{p})$, we get 
\begin{equation}
\left\Vert p_{41}\right\Vert ^{2}+\left\Vert p_{21}\right\Vert
^{2}-\left\Vert p_{42}\right\Vert ^{2}=\left\Vert \hat{p}_{41}\right\Vert
^{2}+\left\Vert \hat{p}_{21}\right\Vert ^{2}-\left\Vert \hat{p}%
_{42}\right\Vert ^{2},  \label{eq:3D-zeta4-tilde-equal-0-base}
\end{equation}%
\begin{align}
& -\left\Vert p_{21}\right\Vert ^{4}+\left\Vert p_{21}\right\Vert
^{2}\left\Vert p_{31}\right\Vert ^{2}+\left\Vert p_{21}\right\Vert
^{2}\left\Vert p_{32}\right\Vert ^{2}+\left\Vert p_{21}\right\Vert
^{2}\left\Vert p_{41}\right\Vert ^{2}  \notag
\label{eq:3D-varphi4-tilde-equal-0-base} \\
& +\left\Vert p_{21}\right\Vert ^{2}\left\Vert p_{42}\right\Vert
^{2}-2\left\Vert p_{43}\right\Vert ^{2}\left\Vert p_{21}\right\Vert
^{2}-\left\Vert p_{31}\right\Vert ^{2}\left\Vert p_{41}\right\Vert ^{2} 
\notag \\
& +\left\Vert p_{31}\right\Vert ^{2}\left\Vert p_{42}\right\Vert
^{2}+\left\Vert p_{32}\right\Vert ^{2}\left\Vert p_{41}\right\Vert
^{2}-\left\Vert p_{32}\right\Vert ^{2}\left\Vert p_{42}\right\Vert ^{2} 
\notag \\
& =-\left\Vert \hat{p}_{21}\right\Vert ^{4}+\left\Vert \hat{p}%
_{21}\right\Vert ^{2}\left\Vert \hat{p}_{31}\right\Vert ^{2}+\left\Vert \hat{%
p}_{21}\right\Vert ^{2}\left\Vert \hat{p}_{32}\right\Vert ^{2}+\left\Vert 
\hat{p}_{21}\right\Vert ^{2}\left\Vert \hat{p}_{41}\right\Vert ^{2}  \notag
\\
& +\left\Vert \hat{p}_{21}\right\Vert ^{2}\left\Vert \hat{p}_{42}\right\Vert
^{2}-2\left\Vert \hat{p}_{43}\right\Vert ^{2}\left\Vert \hat{p}%
_{21}\right\Vert ^{2}-\left\Vert \hat{p}_{31}\right\Vert ^{2}\left\Vert \hat{%
p}_{41}\right\Vert ^{2}  \notag \\
& +\left\Vert \hat{p}_{31}\right\Vert ^{2}\left\Vert \hat{p}_{42}\right\Vert
^{2}+\left\Vert \hat{p}_{32}\right\Vert ^{2}\left\Vert \hat{p}%
_{41}\right\Vert ^{2}-\left\Vert \hat{p}_{32}\right\Vert ^{2}\left\Vert \hat{%
p}_{42}\right\Vert ^{2},
\end{align}
and
\begin{align}
& \left\vert 
\begin{array}{ccccc}
0 & 1 & 1 & 1 & 1 \\ 
1 & 0 & \left\Vert p_{21}\right\Vert ^{2} & \left\Vert p_{31}\right\Vert ^{2}
& \left\Vert p_{41}\right\Vert ^{2} \\ 
1 & \left\Vert p_{21}\right\Vert ^{2} & 0 & \left\Vert p_{32}\right\Vert ^{2}
& \left\Vert p_{42}\right\Vert ^{2} \\ 
1 & \left\Vert p_{31}\right\Vert ^{2} & \left\Vert p_{32}\right\Vert ^{2} & 0
& \left\Vert p_{43}\right\Vert ^{2} \\ 
1 & \left\Vert p_{41}\right\Vert ^{2} & \left\Vert p_{42}\right\Vert ^{2} & 
\left\Vert p_{43}\right\Vert ^{2} & 0%
\end{array}
\right\vert  \notag \\
& =\left\vert 
\begin{array}{ccccc}
0 & 1 & 1 & 1 & 1 \\ 
1 & 0 & \left\Vert \hat{p}_{21}\right\Vert ^{2} & \left\Vert \hat{p}%
_{31}\right\Vert ^{2} & \left\Vert \hat{p}_{41}\right\Vert ^{2} \\ 
1 & \left\Vert \hat{p}_{21}\right\Vert ^{2} & 0 & \left\Vert \hat{p}%
_{32}\right\Vert ^{2} & \left\Vert \hat{p}_{42}\right\Vert ^{2} \\ 
1 & \left\Vert \hat{p}_{31}\right\Vert ^{2} & \left\Vert \hat{p}%
_{32}\right\Vert ^{2} & 0 & \left\Vert \hat{p}_{43}\right\Vert ^{2} \\ 
1 & \left\Vert \hat{p}_{41}\right\Vert ^{2} & \left\Vert \hat{p}%
_{42}\right\Vert ^{2} & \left\Vert \hat{p}_{43}\right\Vert ^{2} & 0%
\end{array}%
\right\vert  \label{eq:3D-vartheta4-tilde-equal-0-base}
\end{align}%
Since we know that $\left\Vert p_{21}\right\Vert =\left\Vert \hat{p}%
_{21}\right\Vert $, $\left\Vert p_{31}\right\Vert =\left\Vert \hat{p}%
_{31}\right\Vert $, and $\left\Vert p_{32}\right\Vert =\left\Vert \hat{p}%
_{32}\right\Vert $, we can use (\ref{eq:3D-zeta4-tilde-equal-0-base}), (\ref%
{eq:3D-varphi4-tilde-equal-0-base}), and (\ref%
{eq:3D-vartheta4-tilde-equal-0-base}) to show $\left\Vert p_{41}\right\Vert
=\left\Vert \hat{p}_{41}\right\Vert $, $\left\Vert p_{42}\right\Vert
=\left\Vert \hat{p}_{42}\right\Vert $, and $\left\Vert p_{43}\right\Vert
=\left\Vert \hat{p}_{43}\right\Vert $.

Repeating the same analysis on $\zeta _{l}$, $\varphi _{l}$, and $\vartheta
_{l}$, $l=5,\ldots ,N$ gives $\mathbf{V}(p)=\mathbf{V}(\hat{p})$ and $\gamma
(p)=\gamma (\hat{p})$. Then, by Lemma \ref{lem:scgt-3d}, we know $F$ and $%
\hat{F}$ are strongly congruent.

\subsection{Lemma \protect\ref{lem:3d-error-variables}}

\label{Sec:proof:error-variables}

Since agents $i$, $j$, and $k$ are located at their desired inter-agent
distances, we can let $p_{i}=\left[ -d_{ji}/2,0,0\right] $, $p_{j}=\left[
d_{ji}/2,0,0\right] $, $p_{k}=\left[ x_{k},y_{k},0\right] $, and $p_{l}=%
\left[ x_{l},y_{l},z_{l}\right] $ without the loss of generality. We also
assume $y_{k}>0$ for simplicity. From the above coordinates, we get 
\begin{equation}
\left\Vert p_{li}\right\Vert ^{2}=\left( x_{l}+\frac{d_{ji}}{2}\right)
^{2}+y_{l}^{2}+z_{l}^{2}  \label{eq:3D-OB:p-li-square}
\end{equation}%
\begin{equation}
\left\Vert p_{lj}\right\Vert ^{2}=\left( x_{l}-\frac{d_{ji}}{2}\right)
^{2}+y_{l}^{2}+z_{l}^{2}.  \label{eq:3D-OB:p-lj-square}
\end{equation}%
After solving for $x_{l}$, we arrive at 
\begin{equation}
x_{l}=\frac{\left\Vert p_{li}\right\Vert ^{2}-\left\Vert p_{lj}\right\Vert
^{2}}{2d_{ji}}.  \label{eq:3d-sol-x}
\end{equation}%
When $\tilde{\zeta}_{l}=0$, we know from (\ref{eq:3D-OB:zeta-l}) and (\ref%
{eq:zeta-error-3d}) that $\left\Vert p_{li}\right\Vert ^{2}-\left\Vert
p_{lj}\right\Vert ^{2}=d_{li}^{2}-d_{lj}^{2}$. Substituting this into (\ref%
{eq:3d-sol-x}) gives 
\begin{equation}
x_{l}=\dfrac{d_{li}^{2}-d_{lj}^{2}}{2d_{ji}}.  \label{xl}
\end{equation}%
This means that all points satisfying $\tilde{\zeta}_{l}=0$ lie on the plane
defined by (\ref{xl}) (blue plane in Figure \ref{fig:ortho-3D-errors}),
which is normal to vector $p_{ji}$ .

Now, substituting the known coordinates of $p_{i}$, $p_{j}$, $p_{k}$, and $%
p_{l}$ into (\ref{eq:3D-OB:varphi-l}) yields 
\begin{equation}
\varphi _{l}=d_{ji}\left\Vert n_{ijk}^{\ast }\right\Vert y_{l}
\label{varphi_l1}
\end{equation}%
where $n_{ijk}^{\ast }:=n_{ijk}(p^{\ast })$ and 
\begin{equation*}
\left\Vert n_{ijk}^{\ast }\right\Vert =\left\{ 
\begin{array}{ll}
\left\Vert p_{ki}^{\ast }\times p_{kj}^{\ast }\right\Vert & \text{if }%
\{i,j,k\}\neq \{1,2,3\} \\ 
1 & \text{if }\{i,j,k\}=\{1,2,3\}.%
\end{array}%
\right.
\end{equation*}%
When $\tilde{\varphi}_{l}=0$, we have from (\ref{varphi_l1}) and (\ref%
{eq:varphi-error-3d}) that 
\begin{equation}
y_{l}=\frac{\varphi _{l}^{\ast }}{d_{ji}\left\Vert n_{ijk}^{\ast
}\right\Vert }.  \label{eq:3d-sol-y}
\end{equation}%
This indicates that all the points satisfying $\tilde{\varphi}_{l}=0$ lie on
the plane defined by (\ref{eq:3d-sol-y}) (red plane in Figure \ref%
{fig:ortho-3D-errors}), which is orthogonal to plane $\tilde{\zeta}_{l}=0$.

Finally, we can use (\ref{eq:3D-OB:vartheta-l}) and (\ref{nstar}) to write%
\begin{equation}
\vartheta _{l}=p_{li}^{\intercal }\frac{n_{ijk}}{\left\Vert
n_{ijk}\right\Vert }\left\Vert n_{ijk}\right\Vert =p_{li}^{\intercal }\frac{%
p_{ki}\times p_{kj}}{\left\Vert p_{ki}\times p_{kj}\right\Vert }\left\Vert
n_{ijk}\right\Vert .  \label{varthetal1}
\end{equation}%
From the known coordinates of $p_{i}$, $p_{j}$, $p_{k}$, and $p_{l}$, we
obtain $p_{li}^{\intercal }\left( p_{ki}\times p_{kj}\right) =-
z_{l}\left\Vert n_{ijk}\right\Vert $. Since $\left\Vert p_{ji}\right\Vert
=d_{ji}$, $\left\Vert p_{ki}\right\Vert =d_{ki}$, and $\left\Vert
p_{kj}\right\Vert =d_{kj}$, we get from (\ref{eq:n}) that $\left\Vert
n_{ijk}\right\Vert =\left\Vert n_{ijk}^{\ast }\right\Vert $. Therefore, (\ref%
{varthetal1}) simplifies to 
\begin{equation}
\vartheta _{l} = - \left\Vert n_{ijk}^{\ast }\right\Vert z_{l}.
\label{varthetal2}
\end{equation}%
When $\tilde{\vartheta}_{l}=0$, we can use (\ref{eq:vartheta-error-3d}) and (%
\ref{varthetal2}) to get 
\begin{equation}
z_{l} = - \frac{\vartheta _{l}^{\ast }}{\left\Vert n_{ijk}^{\ast
}\right\Vert }.  \label{eq:3d-sol-z}
\end{equation}%
That is, all the points satisfying $\tilde{\vartheta}_{l}=0$ are on the
plane defined by (\ref{eq:3d-sol-z}) (green plane in Figure \ref%
{fig:ortho-3D-errors}), which is orthogonal to planes $\tilde{\zeta}_{l}=0$
and $\tilde{\varphi}_{l}=0$. 
\begin{figure}[tbph]
\centering
\adjincludegraphics[scale=0.55, trim={{0.0\width}
			{0.0\height} {0.0\width} {0.00\height}},clip]{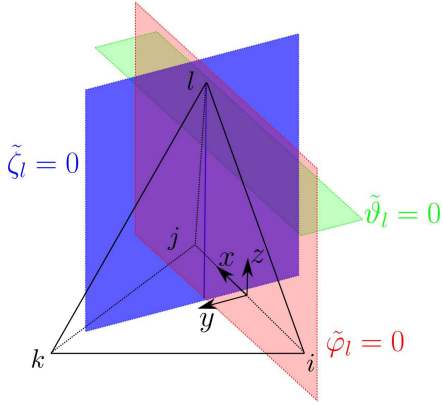}
\caption[Graphical representation of projection errors in 3D.]{Graphical
representation of projection errors: Points on the blue plane satisfy $%
\tilde{\protect\zeta}_{l}=0$; points on the red plane satisfy $\tilde{%
\protect\varphi}_{l}=0$; points on the green plane satisfy $\tilde{\protect%
\vartheta}_{l}=0$.}
\label{fig:ortho-3D-errors}
\end{figure}

\section{Desired Projection Variables}

\label{Sec:desired-projection-variables-3d}

In the following, we show how the desired values for the 3D projection
variables can be computed in terms of the desired inter-agent distances.

\noindent First, we have from (\ref{eq:3D-OB:zeta-l}) that 
\begin{align}
\zeta _{l}^{\ast }=& (p_{li}^{\ast })^{\intercal }p_{ji}^{\ast }  \notag \\
=& \left\Vert p_{li}^{\ast }\right\Vert \left\Vert p_{ji}^{\ast }\right\Vert 
\dfrac{\left\Vert p_{li}^{\ast }\right\Vert ^{2}+\left\Vert p_{ji}^{\ast
}\right\Vert ^{2}-\left\Vert p_{lj}^{\ast }\right\Vert ^{2}}{2\left\Vert
p_{li}^{\ast }\right\Vert \left\Vert p_{ji}^{\ast }\right\Vert }  \notag \\
=& \dfrac{d_{li}^{2}+d_{ji}^{2}-d_{lj}^{2}}{2}.  \label{zeta*_l}
\end{align}

Next, from (\ref{eq:3D-OB:varphi-l}): 
\begin{eqnarray}
\varphi _{l}^{\ast } &=&(p_{li}^{\ast })^{\intercal }(n_{ijk}^{\ast }\times
p_{ji}^{\ast })=(n_{ijk}^{\ast })^{\intercal }\left( p_{ji}^{\ast }\times
p_{li}^{\ast }\right)  \notag \\
&=&\left\Vert n_{ijk}^{\ast }\right\Vert \frac{\left( n_{ijk}^{\ast }\right)
^{\intercal }}{\left\Vert n_{ijk}^{\ast }\right\Vert }\left( p_{ji}^{\ast
}\times p_{li}^{\ast }\right) .  \label{varphi*_l0}
\end{eqnarray}
Notice from (\ref{eq:n}) that 
\begin{equation}
\frac{n_{ijk}^{\ast }}{\left\Vert n_{ijk}^{\ast }\right\Vert }=\frac{%
p_{ki}^{\ast }\times p_{kj}^{\ast }}{\left\Vert p_{ki}^{\ast }\times
p_{kj}^{\ast }\right\Vert }  \label{nstar}
\end{equation}
for both $\{i,j,k\}\neq \{1,2,3\}$ and $\{i,j,k\}=\{1,2,3\}$. After applying
(\ref{nstar}) to (\ref{varphi*_l0}), we obtain 
\begin{equation}
\varphi _{l}^{\ast }=\frac{\left\Vert n_{ijk}^{\ast }\right\Vert }{%
\left\Vert p_{ki}^{\ast }\times p_{kj}^{\ast }\right\Vert }\left(
p_{ki}^{\ast }\times p_{kj}^{\ast }\right) ^{\intercal }\left( p_{ji}^{\ast
}\times p_{li}^{\ast }\right) .  \label{varphi*_l2}
\end{equation}
From (\ref{eq:n}), 
\begin{equation}
\frac{\left\Vert n_{ijk}^{\ast }\right\Vert }{\left\Vert p_{ki}^{\ast
}\times p_{kj}^{\ast }\right\Vert }=\left\{ 
\begin{array}{ll}
1 & \text{if }\{i,j,k\}\neq \{1,2,3\} \\ 
\dfrac{1}{\left\Vert p_{31}^{\ast }\times p_{32}^{\ast }\right\Vert } & 
\text{if }\{i,j,k\}=\{1,2,3\}.%
\end{array}%
\right.  \label{term1}
\end{equation}
Given (\ref{regular area}) and the fact that $\left\Vert p_{ji}^{\ast
}\right\Vert =d_{ji}$, $\forall (j,i)\in \mathcal{E}$, it is obvious that (%
\ref{term1}) is only dependent on the desired distances. Now, 
\begin{eqnarray}
&&\left( p_{ki}^{\ast }\times p_{kj}^{\ast }\right) ^{\intercal }\left(
p_{ji}^{\ast }\times p_{li}^{\ast }\right)  \notag \\
&=&\left( -p_{kj}^{\ast }\times p_{ki}^{\ast }\right) ^{\intercal }\left(
-p_{li}^{\ast }\times p_{ji}^{\ast }\right)  \notag \\
&=&\left( -p_{kj}^{\ast }\times p_{kj}^{\ast }-p_{kj}^{\ast }\times
p_{ji}^{\ast }\right) ^{\intercal }\left( -p_{lj}^{\ast }\times p_{ji}^{\ast
}-p_{ji}^{\ast }\times p_{ji}^{\ast }\right)  \notag \\
&=&\left( p_{kj}^{\ast }\times p_{ij}^{\ast }\right) ^{\intercal }\left(
p_{lj}^{\ast }\times p_{ij}^{\ast }\right)  \notag \\
&=&\left\Vert p_{kj}^{\ast }\right\Vert \left\Vert p_{lj}^{\ast }\right\Vert 
\dfrac{\left\Vert p_{kj}^{\ast }\right\Vert ^{2}+\left\Vert p_{lj}^{\ast
}\right\Vert ^{2}-\left\Vert p_{lk}^{\ast }\right\Vert ^{2}}{2\left\Vert
p_{kj}^{\ast }\right\Vert \left\Vert p_{lj}^{\ast }\right\Vert }\left\Vert
p_{ij}^{\ast }\right\Vert ^{2}  \notag \\
&&-\left\Vert p_{kj}^{\ast }\right\Vert \left\Vert p_{ij}^{\ast }\right\Vert 
\dfrac{\left\Vert p_{kj}^{\ast }\right\Vert ^{2}+\left\Vert p_{ij}^{\ast
}\right\Vert ^{2}-\left\Vert p_{ik}^{\ast }\right\Vert ^{2}}{2\left\Vert
p_{kj}^{\ast }\right\Vert \left\Vert p_{ij}^{\ast }\right\Vert }  \notag \\
&\cdot &\left\Vert p_{ij}^{\ast }\right\Vert \left\Vert p_{lj}^{\ast
}\right\Vert \dfrac{\left\Vert p_{ij}^{\ast }\right\Vert ^{2}+\left\Vert
p_{lj}^{\ast }\right\Vert ^{2}-\left\Vert p_{li}^{\ast }\right\Vert ^{2}}{%
2\left\Vert p_{ij}^{\ast }\right\Vert \left\Vert p_{lj}^{\ast }\right\Vert },
\label{term2}
\end{eqnarray}%
which is also only dependent on the desired distances. Thus, (\ref%
{varphi*_l2}) is only a function of $d_{ji}$, $\forall (j,i)\in \mathcal{E}$.

Finally, it follows from (\ref{eq:3D-OB:vartheta-l}) that 
\begin{equation}
\vartheta _{l}^{\ast }=(p_{li}^{\ast })^{\intercal }n_{ijk}^{\ast
}=(p_{li}^{\ast })^{\intercal }n_{ijk}^{\ast }-(p_{ki}^{\ast })^{\intercal
}n_{ijk}^{\ast }=(p_{lk}^{\ast })^{\intercal }n_{ijk}^{\ast }.
\label{vartheta*_l0}
\end{equation}%
If $\{i,j,k\}\neq \{1,2,3\}$, then from (\ref{eq:n}) 
\begin{align}
\vartheta _{l}^{\ast }=& (p_{lk}^{\ast })^{\intercal }\left( p_{ki}^{\ast
}\times p_{kj}^{\ast }\right)  \notag \\
=& (p_{lk}^{\ast })^{\intercal }\left[ \left( p_{li}^{\ast }-p_{lk}^{\ast
}\right) \times \left( p_{lj}^{\ast }-p_{lk}^{\ast }\right) \right]  \notag
\\
=& (p_{lk}^{\ast })^{\intercal }\left( p_{li}^{\ast }\times p_{lj}^{\ast
}-p_{li}^{\ast }\times p_{lk}^{\ast }-p_{lk}^{\ast }\times p_{lj}^{\ast
}+p_{lk}^{\ast }\times p_{lk}^{\ast }\right)  \notag \\
=& (p_{lk}^{\ast })^{\intercal }\left( p_{li}^{\ast }\times p_{lj}^{\ast
}\right) =-6V_{ijkl}^{\ast }  \label{vartheta*_l1}
\end{align}%
where $V_{ijkl}^{\ast }:=V(p_{i}^{\ast },p_{j}^{\ast },p_{k}^{\ast
},p_{l}^{\ast })$ from (\ref{eq:signed-volume}). Note that $V_{ijkl}^{\ast }$
can be calculated using (\ref{eq:V*}) with $\left\Vert p_{ji}\right\Vert
=d_{ji}$ where the sign is based on the desired ordering of vertices $i,j,k$
per the convention described in Section \ref{Sec: Strong Congr}. If $%
\{i,j,k\}=\{1,2,3\}$, then from (\ref{eq:n}) and the calculations in (\ref%
{vartheta*_l1}), we arrive at 
\begin{equation}
\vartheta _{l}^{\ast }=-\frac{6V_{123l}^{\ast }}{\left\Vert p_{31}^{\ast
}\times p_{32}^{\ast }\right\Vert }  \label{vartheta*_l2}
\end{equation}%
where (\ref{regular area}) is used to calculate the denominator in terms of
the desired distances.

The desired projections $\zeta _{2}^{\ast }$, $\zeta _{3}^{\ast }$, and $%
\varphi _{3}^{\ast }$ are special cases of the above variables. From (\ref%
{eq:zeta-2}), we have that 
\begin{equation}
\zeta _{2}^{\ast }=\left( p_{21}^{\ast }\right) ^{\intercal }\dfrac{%
p_{21}^{\ast }}{\left\Vert p_{21}^{\ast }\right\Vert }=\left\Vert
p_{21}^{\ast }\right\Vert =d_{21}.
\end{equation}
From (\ref{eq:zeta-3}) and (\ref{zeta*_l}), we obtain 
\begin{equation*}
\zeta _{3}^{\ast }=(p_{31}^{\ast })^{\intercal }p_{21}^{\ast }=\dfrac{%
d_{31}^{2}+d_{21}^{2}-d_{32}^{2}}{2}.
\end{equation*}
Finally, from (\ref{eq:varphi-3}), we have 
\begin{align}
\varphi _{3}^{\ast }=& \left( p_{31}^{\ast }\right) ^{\intercal }\left(
n_{123}^{\ast }\times p_{21}^{\ast }\right) 
= \left( n_{123}^{\ast }\right) ^{\intercal }\left( p_{31}^{\ast }\times
p_{32}^{\ast }\right)  \notag \\
=& \left\Vert p_{31}^{\ast }\times p_{32}^{\ast }\right\Vert 
= 2\breve{S}_{123}^{\ast }  \label{eq:varphi-3d}
\end{align}%
where $\breve{S}_{123}^{\ast }:=\breve{S}_{123}(p_{31}^{\ast}, p_{32}^{\ast})$ from (\ref{regular area}) is only dependent on the desired distances.

\end{appendices}

\bibliographystyle{dcu}
\bibliography{References}

\end{document}